\documentclass{comsoc2023}[a4paper]
\usepackage{amsthm}

\usepackage{graphicx}
\usepackage[ruled,linesnumbered]{algorithm2e} 
\usepackage{amssymb,amsmath}
\DeclareMathOperator*{\argmin}{arg\,min}

\title{Strategic Voting in the Context of Stable-Matching of Teams}
\author{Leora Schmerler and Noam Hazon and Sarit Kraus}
\begin{document}


\renewcommand{\algorithmcfname}{ALGORITHM}
\SetAlFnt{\small}
\SetAlCapFnt{\small}
\SetAlCapNameFnt{\small}
\SetAlCapHSkip{0pt}
\IncMargin{-\parindent}
\newcommand{\calF}{{\mathcal F}}
\newcommand{\calL}{{\mathcal L}}
\newcommand{\pos}[2]{{pos(#1,#2)}}
\newcommand{\todo}[1]{{\bf [{** Todo ** \small #1}]}}
\def\set#1{\{#1\}}
\newcommand{\shira}[1]{\textcolor{green}{[Shira: #1]}}

\newenvironment{Function}[1][htb]
  {\renewcommand{\algorithmcfname}{Procedure}
  \begin{algorithm}[#1]%
  
  }{\end{algorithm}}

\renewcommand{\algorithmcfname}{ALGORITHM}

\newtheorem{example}{Example}
\newtheorem{theorem}{Theorem}
\newtheorem{definition}{Definition}
\newtheorem{lemma}{Lemma}
\newtheorem{corollary}{Corollary}





\begin{abstract}
In the celebrated stable-matching problem, there are two sets of agents M and W, and the members of M only have preferences over the members of W and vice versa. It is usually assumed that each member of M and W is a single entity. However, there are many cases in which each member of M or W represents a team that consists of several individuals with common interests.
For example, students may need to be matched to professors for their final projects, but each project is carried out by a team of students. Thus, the students first form teams, and the matching is between teams of students and professors.  

When a team is considered as an agent from M or W, it needs to have a preference order that represents it. A voting rule is a natural mechanism for aggregating the preferences of the team members into a single preference order.
In this paper, we investigate the problem of strategic voting in the context of stable-matching of teams. Specifically, we assume that members of each team use the Borda rule for generating the preference order of the team. 
Then, the Gale-Shapley algorithm is used for finding a stable-matching, where the set M is the proposing side. 
We show that the single-voter manipulation problem can be solved in polynomial time, both when the team is from M and when it is from W. We show that the coalitional manipulation problem is computationally hard, but it can be solved approximately both when the team is from M and when it is from W. 
\end{abstract}



\section{Introduction}
%
Matching is the process in which agents from different sets are matched with each other. 
The theory of matching originated with the seminal work of Gale and Shapley \cite{gale1962college}, and since then intensive research has been conducted in this field. 
Notably, the theory of matching has also been successfully applied to many real-world applications including college admissions and school matching \cite{abdulkadirouglu2005new}, matching residents to hospitals \cite{roth1996nrmp},
and kidney exchange \cite{roth2004kidney}. 
A very common matching problem, which is also the problem that was studied by Gale and Shapley in their original paper, is the \textit{stable-matching} problem. In this problem there are two equally sized disjoint sets of agents, $M$ and $W$, and the members of $M$ have preferences over only the members of $W$, and vice versa. The goal is to find a stable bijection (i.e., matching) from the agents of $M$ to the agents of $W$, where the stability requirement is that no pair of agents prefers a match with each other over their matched
partners.
Many works have analyzed this setting, and they assume that each member of the sets $M$ and $W$ represents a single agent. However, there are many cases in which each member of $M$ or $W$ represents more than one individual \cite{perach2022stable}.

For example, suppose that teams of students need to be matched with professors who will serve as their advisors in their final projects. It is common that students form their teams based on friendship connections and common interests and then approach the professors. Therefore, each team is considered to be a single agent for the matching process: the professors may have different preferences regarding which team they would like to mentor, and the teams may have preferences regarding which professor they would like as their mentor. Clearly, even though the team is considered to be a single agent for the matching process, it is still composed of several students, and they may have different opinions regarding the appropriate mentor for their team. Thus, every team needs a mechanism that aggregates the students' opinions and outputs a single preference order that represents the team for the matching process, and a voting rule is a natural candidate.

Indeed, voters might benefit from reporting rankings different from their true ones, and this problem of manipulation also exists in the context of matching.
For example, suppose that there are $4$ possible professors, denoted by $p_1,p_2,p_3$ and $p_4$ and $4$ teams. Now, suppose that one of the students, denoted $r$, who is a member of one of the teams, prefers $p_1$ over $p_2$, $p_2$ over $p_3$, and $p_3$ over $p_4$. It is possible that $r$ will gain an (unauthorized) access to the preferences of the professors and to the preferences of the other teams. Since the matching algorithm is usually publicly known, $r$ might be able to reason that $p_3$ is matched with his team, but if $r$ votes strategically and misreports his preferences then $p_2$ will be matched with his team.

In this paper, we investigate the problem of strategic voting in the context of stable-matching of teams. We assume that the members of each team use the Borda rule as a \emph{social welfare function (SWF)}, which outputs a complete preference order. This preference order represents the team for the matching process. The agents then use the Gale-Shapley (GS) algorithm for finding a stable-matching.
In the GS algorithm, one set of agents makes proposals to the other set of agents, and it is assumed that $M$ is the proposing side and $W$ is the proposed-to side. The proposing side and proposed-to side are commonly
referred to as men and women, respectively.
Note that the GS algorithm treats the men and women differently.
Therefore, every manipulation problem in the context of stable-matching has two variants: one in which the teams are from the men's side, and another one in which the teams are from the women's side.
Moreover, we analyze both manipulation by a single voter and coalitional manipulation. In a single voter manipulation, the goal is to find a preference order for a single manipulator such that his team will be matched by the GS algorithm with a specific preferred agent. In the coalitional manipulation setting, there are several voters who collude and coordinate their votes so that an agreed upon agent will be matched with their team.

We begin by studying manipulation from the men's side, and show that the single voter manipulation problem can be solved in polynomial time. We then analyze the coalitional manipulation problem, and show that the problem is computationally hard. However, we provide a polynomial-time algorithm with the following guarantee: given a manipulable instance with $|R|$ manipulators, the algorithm finds a successful manipulation with at most one additional manipulator. 
We then study manipulation from the women's side. Manipulation here is more involved, and we propose different algorithms, but with almost the same computational complexity as in manipulation from the men's side. That is, the single voter manipulation problem can be solved in polynomial time, and the coalitional manipulation problem is computationally hard. Indeed, we provide a polynomial-time algorithm with the following guarantee: given a manipulable instance with $|R|$ manipulators, the algorithm finds a successful manipulation with at most two additional manipulators.

The contribution of this work is twofold. First, it provides an analysis of a voting manipulation in the context of stable-matching of teams, a problem that has not been investigated to date. Second, our work concerns the manipulation of Borda as an SWF, which has scarcely been investigated.
%
%
%
\section{Related Work}
The computational analysis of voting manipulation has been vastly studied in different settings. We refer the reader to the survey provided by Faliszewski and Procaccia \cite{faliszewski2010ai}, and the more recent survey by Conitzer and Walsh \cite{conitzer2016barriers}. However, most of the works on voting manipulation analyze the problem with no actual context, and where a voting rule
is used to output one winning candidate or a set of tied winning candidates (i.e., a social choice function).
In this work, we investigate manipulation of Borda as a SWF, which outputs a complete preference order of the candidates, and analyze it within the context of stable-matching.

Indeed, there are a few papers that investigate the manipulation of SWFs. The first work that directly deals with the manipulation of SWF was by Bossert and Storcken \cite{bossert1992strategy}, who assumed that a voter prefers one order over another if the former is closer to her own preferences than the latter according to the Kemeny distance. 
Bossert and Sprumont \cite{bossert2014strategy}
assumed that a voter prefers one order over another if the former is strictly between the latter and the voter's own preferences. Built on this definition, their work studies three classes of SWF that are not prone to manipulation (i.e., strategy-proof).
Dogan and Lainé \cite{dogan2016strategic} 
characterized the conditions to be imposed on SWFs so that if we extend the preferences of the voters to preferences over orders in specific ways, the SWFs will not be prone to manipulation. Our work also investigates the manipulation of SWF, but we analyze the SWF in the specific context of stable-matching. Therefore, unlike all of the above works, the preferences of the manipulators are well-defined and no additional assumptions are needed. The work that is closest to ours is that of Schmerler and Hazon \cite{schmerler2021strategic}. They assume that a positional scoring rule is used as a SWF, and study the manipulation of the SWF in the context of negotiation. 

The strategic aspects of the GS algorithm have previously been studied in the literature. 
It was first shown that reporting the true preferences is a weakly dominant strategy for men, but women may have an incentive to misreport their preferences \cite{dubins1981machiavelli,roth1982economics,gale1985ms}.
Teo et al. \cite{teo2001gale} provided a polynomial-time algorithm for computing the optimal manipulation by a woman.
Shen et al. \cite{shen2021coalitional} generalized this result to manipulation by a coalition of women.
For the proposing side, Dubins and Freedman \cite{dubins1981machiavelli} investigated the strategic actions of a coalition of men, and proved that there is no manipulation that is a strict improvement for every member of the coalition. Huang \cite{huang2006cheating} studied manipulation that is a weak improvement for every member of a coalition of men.
Hosseini et al. \cite{hosseiniaccomplice} introduced a new type of strategic action: manipulation through an accomplice. In this manipulation, a man misreports his preferences in behalf of a woman, and Hosseini et al. provided a polynomial time algorithm for computing an optimal accomplice manipulation, and they further generalized this model in \cite{hosseini2022two}.
All of these works consider the manipulation of the GS algorithm, while we study the manipulation of Borda as a SWF. Indeed, the output of the SWF is used (as part of the input) for the GS algorithm.
As an alternative to the GS algorithm, Pini et al. \cite{pini2011manipulation} show how voting rules which are NP-hard to manipulate can be used to build stable-matching procedures, which are themselves NP-hard to manipulate.


\section{Preliminaries} \label{sec:prem}
We assume that there are two equally sized disjoint sets of agents, $M$ and $W$. Let $k=|M|=|W|$. The members of $M$ have preferences over only the members of $W$, and vice versa. The preference of each $m \in M$, denoted by $\succ_m$, is a strict total order over the agents in $W$. The preference profile $\succ_M$ is a vector $(\succ_{m_1},\succ_{m_2},\ldots,\succ_{m_k})$. The preference order $\succ_w$ and the preference profile $\succ_W$ are defined analogously. We will refer to the agents of $M$ as men and to the agents of $W$ as women. 

A matching is a mapping $\mu:M \cup W \rightarrow M \cup W$, such that $\mu(m) \in W$ for all $m \in M$, $\mu(w) \in M$ for all $w \in W$, and $\mu(m) = w$ if and only if $\mu(w) = m$. A stable-matching is a matching in which there is no blocking pair. That is, there is no man $m$ and woman $w$ such that $w \succ_m \mu(m)$ and $m \succ_w \mu(w)$. 
The GS algorithm finds a stable-matching, and it works as follows. There are multiple rounds, and each round is composed of a proposal phase followed by a rejection phase.
In a proposal phase, each unmatched man proposes to his favorite woman from among those who have not yet rejected him (regardless of whether the woman is already matched). In the rejection phase, each woman tentatively accepts her favorite proposal and rejects all of the other proposals. The
algorithm terminates when no further proposals can be made.
Let $o(w)$ be the set of men that proposed to $w$ in one of the rounds of the GS algorithm.

In our setting, (at least) one of the agents of $M$ ($W$) is a team that runs an election for determining its preferences.
That is, there is a man $\hat{m}$ (woman $\hat{w}$), which is associated with a set of voters, $V$. The preference of each $v \in V$, denoted by $\ell_v$, is a strict total order over $W$ ($M$). The preference profile ${\calL}$ is a vector of the preference orders for each $v \in V$. 
%
The voters use the \textit{Borda} rule as a SWF, denoted by $\calF$, which is a mapping of the set of all preference profiles to a single strict preference order. Specifically, in the Borda rule, each voter $v$ awards the candidate that is placed in the top-most position in $\ell_v$ a score of $k-1$, the candidate in the second-highest position in $\ell_v$ a score of $k-2$, etc.
Then, for the output of $\calF$, the candidate with the highest aggregated score is placed in the top-most position, the candidate with the second-highest score is placed in the second-highest position, etc. 
Since ties are possible, we assume that a lexicographical tie-breaking rule is used. Note that the output of $\calF$ is the preference order of $\hat{m}$ ($\hat{w}$). That is, $\succ_{\hat{m}} = \calF(\calL)$, and $\succ_{\hat{w}}$ is defined analogously. 

Recall that the GS algorithm finds a stable matching, given $\succ_M$ and $\succ_W$. Given a man $m \in M$, let $\succ_{M-m}$ be the preference profile of all of the men besides $m$, and $\succ_{W-w}$ is defined analogously. We consider a setting in which the input for the GS algorithm is $\succ_{M-\hat{m}}, \succ_{\hat{m}}$, and $\succ_W$, and thus $\mu(\hat{m})$ is the spouse that is the match of $\hat{m}$ according to the output of the GS algorithm. We also consider a setting in which the input for the GS algorithm is $\succ_{W-\hat{w}},\succ_{\hat{w}}$ and $\succ_M$, and thus $\mu(\hat{w})$ is the spouse that is the match of $\hat{w}$ according to the output of the GS algorithm. 
In some circumstances, we would like to examine the output of the GS algorithm for different possible preference orders that represent a man $m \in M$.
We denote by $\mu_x(m,\succ)$ the spouse that is the match of $m$ when the input for the GS algorithm is $\succ_{M-m}$, $\succ$ (instead of $\succ_{m}$), and $\succ_W$. We define $\mu_x(w,\succ)$ and $o_x(w,\succ)$ similarly.

We study the setting in which there exists a manipulator $r$ among the voters associated with a man $\hat{m}$ (woman $\hat{w}$), and his (her) preference order is $\ell_r$. The preference order that represents $\hat{m}$ ($\hat{w}$) is thus $\calF(\calL \cup \set{\ell_r})$.
We also study the setting in which there is a set of $R$ manipulators, their preference profile is $\calL_R = \set{\ell_{1},\ell_{2},\ldots,\ell_{|R|}}$, and the preference order that represents $\hat{m}$ ($\hat{w}$) is thus $\calF(\calL \cup \calL_R)$.
For clarity purposes we slightly abuse notation, and write $\mu(\hat{m},\ell_r)$ for denoting the spouse that is the match of $\hat{m}$ according to the output of the GS algorithm, given that its input is $\succ_{M-\hat{m}}, \calF(\calL \cup \set{\ell_r})$, and $\succ_W$.
We define $\mu(\hat{w},\ell_r)$, $o(\hat{w},\ell_r)$, $\mu(\hat{m},\calL_R)$, $\mu(\hat{w},\calL_R)$ and $o(\hat{w},\calL_R)$ similarly.

Let $s(c,\ell_v)$ be the score of candidate $c$ from $\ell_v$. 
Similarly, let $s(c,\calL)$ be the total score of candidate $c$ from $\calL$, i.e., $s(c,\calL) = \sum_{v \in V} s(c,\ell_v)$. Similarly, $s(c,\calL,\ell_r) = \sum_{v \in V} s(c,\ell_v) + s(c,\ell_r)$, and \sloppy{$s(c,\calL,\calL_R) = \sum_{v \in V} s(c,\ell_v) + \sum_{r \in R}s(c,\ell_r)$}. 
Since we use a lexicographical tie-breaking rule, we write  that $(c, \ell) > (c', \ell')$ if $s(c, \ell) > s(c', \ell')$ or $s(c, \ell) = s(c', \ell')$ but $c$ is preferred over $c'$ according to the lexicographical tie-breaking rule. We define $(c, \calL, \ell) > (c', \calL, \ell')$ and $(c, \calL, \calL_R) > (c', \calL, \calL'_R)$ similarly.
Note that due to space constraints, almost all of the proofs are deferred to the full version of the paper~\cite{schmerler2022strategic}.


\section{Men's Side}
We begin by considering the variant in which a specific voter, or a coalition of voters, are associated with an agent $\hat{m}$, and they would like to manipulate the election so that a preferred spouse $w^*$ will be the match of $\hat{m}$.
\subsection{Single Manipulator} \label{sec:single_manipulator}
With a single manipulator, the Manipulation in the context of Matching from the Men's side (MnM-m) is defined as follows:

\begin{definition}[MnM-m]
    We are given a man $\hat{m}$, the preference profile $\calL$ of the honest voters that associate with $\hat{m}$, the preference profile $\succ_{M-\hat{m}}$, the preference profile $\succ_{W}$, a specific manipulator $r$, and a preferred woman $w^* \in W$. We are asked whether a preference order $\ell_r$ exists such that $\mu(\hat{m}, \ell_r) = w^*$.
\end{definition}

We show that MnM-m can be decided in polynomial time by  Algorithm~\ref{alg:man_single_manip}, which works as follows. %
The algorithm begins by verifying that a preference order exists for $\hat{m}$, which makes $w^*$ the match of $\hat{m}$. It thus iteratively builds a temporary preference order for $\hat{m}$, $\succ_x$ in lines~\ref{line:begin_single_find_B}-\ref{line:end_single_find_B}. Moreover, during the iterations in lines~\ref{line:begin_single_find_B}-\ref{line:end_single_find_B} the algorithm identifies a set $B$, which is the set of women that might prevent $w^*$ from being $\hat{m}$'s match. Specifically, $\succ_x$, is  initialized as the original preference order of $\hat{m}$, $\succ_{\hat{m}}$. In each iteration, the algorithm finds the woman $b$, which is the match of $\hat{m}$ given that $\succ_x$ is the preference order of $\hat{m}$. If $b$ is placed higher than $w^*$ in $\succ_x$, then $b$ is added to the set $B$, it is placed in $\succ_x$ immediately below $w^*$, and the algorithm proceeds to the next iteration (using the updated $\succ_x$).

Now, if $b=\mu_x(\hat{m},\succ_x)$ is positioned lower than $w^*$ in $\succ_x$, then no preference order exists that makes $w^*$ the match of $\hat{m}$, and the algorithm returns false. If $b=w^*$, then the algorithm proceeds to build the preference order for the manipulator, $\ell_r$. Clearly, $w^*$ is placed in the top-most position in $\ell_r$. Then, the algorithm places all the women that are not in $B$ in the highest available positions. Finally, the algorithm places all the women from $B$ in the lowest positions in $\ell_r$, and they are placed in a reverse order with regard to their order in $\calF(\calL)$.

\begin{algorithm}[hbpt]
\caption{Manipulation by a single voter from the men's side}
\label{alg:man_single_manip}
\SetAlgoLined
\SetAlgoNoEnd
    $B \leftarrow \emptyset$ \\  \label{line:begin_find_B}
    set $\succ_x$ to be $\succ_{\hat{m}}$ \\
    $b \leftarrow \mu_x(\hat{m}, \succ_x)$ \\
    \While{$b \succ_x w^*$}{ \label{line:begin_single_find_B}
        add $b$ to $B$ \\
        move $b$ in $\succ_x$ immediately below $w^*$ \\ \label{line:move_b} 
        $b \leftarrow \mu_x(\hat{m}, \succ_x)$\\
    } \label{line:end_single_find_B}
    \If{$b \neq w^*$}{ \label{line:begin_if_possible}
            \textbf{return} false \label{line:single_man_no_preference}
    } \label{line:end_if_possible}
    \tcp{phase 1:}
    $\ell_r \leftarrow$ empty preference order \\
    place $w^*$ in the highest position in $\ell_r$ \\ \label{line:placing_w*}
    \For{\textbf{each } $w \in W \setminus (B \cup \set{w^*})$}    { \label{line:begin_single_man_place_high}
        place $w$ in the next highest available position in $\ell_r$ \\
    } \label{line:end_single_man_place_high}
    \tcp{phase 2:}
    \While{$B \neq \emptyset$} { \label{line:begin_single_man_reverse}
        $b \leftarrow$ the least preferred woman from $B$ according to $\calF(\calL)$ \\
        place $b$ in the highest available position in $\ell_r$ \\
        remove $b$ from $B$ 
    } \label{line:end_single_man_reverse}
    \If{$\mu(\hat{m}, \ell_r) = w^*$}{ \label{line:single_man_if_match}
        \textbf{return} $\ell_r$ \label{line:single_man_success}
    }
    
    \textbf{return} false
\end{algorithm}

For proving the correctness of Algorithm~\ref{alg:man_single_manip} we use the following known results:

\begin{theorem}[due to~\cite{roth1982economics}]\label{thm:truthful}
In the Gale-Shapley matching procedure which always yields the optimal stable outcome for the set of the men agents, $M$, truthful revelation is a dominant strategy for all the agents in that set. 
\end{theorem}

\begin{lemma}[due to~\cite{huang2006cheating}]\label{lemma:permutation}
\sloppy{For man $m$, his preference list is composed of $(P_L(m),\mu(m),P_R(m))$, where $P_L(m)$ and $P_R(m)$ are respectively those women ranking higher and lower than $\mu(m)$.} Let $A \subseteq W$ and let $\pi_r(A)$ be a random permutation from all $|A|!$ sets.
For a subset of men $S \subseteq M$, if every member $m \in S$ submits a falsified list of the form $(\pi_r(P_L(m)), \mu(m), \pi_r(P_R(m)))$, then $\mu(m)$ stays $m$'s match.
\end{lemma}

An immediate corollary of Theorem~\ref{thm:truthful} is the following.
\begin{corollary} \label{corollary:placing_higher}
Given a man $m$ with his preference order $\succ_m$, let $w_m = \mu(m)$. Let $\succ'_m$ be a preference order for $m$ such that if $w_m \succ_m w$ then $w_m \succ'_m w$. Then, $w_m$ is also the match of $m$ with the preference order $\succ'_m$, i.e., $\mu_x(m,\succ'_m)=w_m$.
\end{corollary}

We begin the analysis of Algorithm~\ref{alg:man_single_manip} by showing that it is possible to verify (in polynomial time) whether a preference order exists for $\hat{m}$, which makes $w^*$ the match of $\hat{m}$. 
We do so by showing that it is sufficient to check whether $w^*= \mu_x(\hat{m},\succ_x)$, where $\succ_x$ is the preference order that is built by Algorithm~\ref{alg:man_single_manip} in lines~\ref{line:begin_single_find_B}-\ref{line:end_single_find_B}.

\begin{lemma}
\label{lemma:exists_pref_at_all}
A preference order $\succ_t$ for $\hat{m}$ exists such that $w^*= \mu_x(\hat{m},\succ_t)$ if and only if 
$w^*= \mu_x(\hat{m},\succ_x)$. 
\end{lemma}

That is, if Algorithm~\ref{alg:man_single_manip} returns false in line~\ref{line:single_man_no_preference} then there is no preference order for $\hat{m}$ that makes $w^*$ the match of $\hat{m}$ (and thus no manipulation is possible for $r$). 

We now show that the set $B$, which is identified by the algorithm in lines~\ref{line:begin_single_find_B}-\ref{line:end_single_find_B}, is a set of woman that might prevent $w^*$ from being $\hat{m}$'s match.

\begin{lemma}\label{lemma:b_prevent_m*}
Given a preference order $\succ_t$ for $\hat{m}$, if there exists $b \in B$ such that $b \succ_t w^*$ then $\mu_x(\hat{m}, \succ_t) \neq w^*$. 
\end{lemma}

That is, Algorithm~\ref{alg:man_single_manip} should place the women from $B$ in the lowest position in $\ell_r$ and $w^*$ in the highest position in $\ell_r$, so that $w^*$ will be preferred over every woman $b \in B$ in $\calF(\calL \cup \set{\ell_r})$.


%


\begin{theorem}
\label{thm:alg_1_correct}
Algorithm~\ref{alg:man_single_manip} correctly decides the MnM-m problem in polynomial time.
\end{theorem}

\subsection{Coalitional Manipulation}
We now study manipulation by a coalition of voters. 
The coalitional manipulation in the context of matching from the men's side is defined as follows:

\begin{definition} [coalitional MnM-m] \label{def:coalition_among_a_man}
    We are given a man $\hat{m}$, the preference profile $\calL$ of the honest voters that associate with $\hat{m}$, the preference profile $\succ_{M-\hat{m}}$, the preference profile $\succ_{W}$, a coalition of manipulators $R$, and a preferred woman $w^* \in W$. We are asked whether a preference profile $\calL_R$ exists such that $\mu(\hat{m}, \calL_R) = w^*$.
\end{definition}

We show that the coalitional MnM-m problem is computationally hard, even with two manipulators. The reduction is from the Permutation Sum problem (as defined by Davies et al.~\cite{davies2011complexity}) that is $NP$-complete~\cite{yu2004minimizing}.

\begin{definition}[Permutation Sum]\label{def:permutationSum} Given $q$ integers $X_1 \leq \ldots \leq X_q$ where $\sum_{i=1}^{q}X_i = q(q+1)$, do two permutations $\sigma$ and $\pi$ of $1$ to $q$ exist such that $\sigma(i)+\pi(i)=X_i$ for all $1 \leq i \leq q$?
\end{definition}

\begin{theorem} \label{thm:coalitional_mnm-m}
Coalitional MnM-m is {NP}-Complete. 
\end{theorem}
%
%
%
%
%
Even though coalitional MnM-m is $NP$-complete, it might still be possible to develop an efficient heuristic algorithm that finds a successful coalitional manipulation.
We use Algorithm~\ref{alg:coalition_manip_man}, which is a generalization of Algorithm~\ref{alg:man_single_manip}, that works as follows.
\begin{algorithm}[hbtp]
\caption{Manipulation by a coalition of voters from the men's side}
\label{alg:coalition_manip_man}
\SetAlgoLined
\LinesNumbered
\SetAlgoNoEnd
    $B \leftarrow \emptyset$ \\  \label{line:coalition_begin_find_B}
    set $\succ_x$ to be $\succ_{\hat{m}}$ \\
    $b \leftarrow \mu_x(\hat{m}, \succ_x)$ \\
    \While{$b \succ_x w^*$}{
        add $b$ to $B$ \\
        place $b$ in $\succ_x$ immediately below $w^*$ \\
        $b \leftarrow \mu_x(\hat{m}, \succ_x)$\\
    } \label{line:coalition_end_find_B}
    \If{$b \neq w^*$}{ \label{line:coalition_begin_if_possible}
        \textbf{return} false
    } \label{line:coalition_end_if_possible}
    \For{\textbf{each} $r \in R$} {\label{line:cmnm-m_begin-stage}
        $\ell_r \leftarrow$ empty preference order \\
        place $w^*$ in the highest position in $\ell_r$ \\
        \For{\textbf{each } $w \in W \setminus (B \cup \set{w^*})$}    { \label{line:coalition_begin_single_man_place_high}
            place $w$ in the next highest available position in $\ell_r$ \\
        } \label{line:coalition_end_single_man_place_high}
        $B' \leftarrow B$ \\
        \While{$B' \neq \emptyset$} { \label{line:begin_coalition_man_reverse}
        $b \leftarrow$ the least preferred woman from $B'$ according to $\calF(\calL \cup \calL_R)$ \\
        place $b$ in the highest available position in $\ell_r$ \\
        remove $b$ from $B'$ 
    } \label{line:end_coalition_man_reverse}
        add $\ell_{r}$ to $\calL_R$
    } \label{line:cmnm-m_end-stage}
    \If{$\mu(\hat{m}) = w^*$}{
            \textbf{return} $\calL_R$
        }
    \textbf{return} false
\end{algorithm} 
Similar to Algorithm~\ref{alg:man_single_manip}, Algorithm~\ref{alg:coalition_manip_man} identifies a set $B$, which is the set of women that might prevent $w^*$ from being $\hat{m}$'s match. In addition, it verifies that a preference order for $\hat{m}$ exists, which makes $w^*$ the match of $\hat{m}$.
Then, Algorithm~\ref{alg:coalition_manip_man} proceeds to build the preference order of every manipulator $r \in R$ similarly to how Algorithm~\ref{alg:man_single_manip} builds the preference order for the single manipulator. Indeed, Algorithm~\ref{alg:coalition_manip_man} builds the preference order of each manipulator $r$ in turn, and the order in which the women in $B$ are placed depends on their order according to $\calF(\calL \cup \calL_R)$. That is, the order in which the woman in $B$ are placed in each $\ell_r$ is not the same for each $r$, since $\calL_R$ is updated in each iteration. We refer to each of the iterations in Lines~\ref{line:cmnm-m_begin-stage}-\ref{line:cmnm-m_end-stage} as a \emph{stage} of the algorithm.
We now show that Algorithm~\ref{alg:coalition_manip_man} is an efficient heuristic that also has a theoretical guarantee. Specifically, the algorithm is guaranteed to find a coalitional manipulation in many instances, and we characterize the instances in which it may fail.
Formally,
\begin{theorem} \label{thm:coalition_manip_man}
Given an instance of coalitional MnM-m,
    \begin{enumerate}
        \item If there is no preference profile making $w^*$ the match of $\hat{m}$, then Algorithm~\ref{alg:coalition_manip_man} will return false. 
        \item If a preference profile making $w^*$ the match of $\hat{m}$ exists, then for the same instance with one additional manipulator, Algorithm~\ref{alg:coalition_manip_man} will return a preference profile that makes $w^*$ the match of $\hat{m}$.
    \end{enumerate}
\end{theorem}
That is, Algorithm~\ref{alg:coalition_manip_man} will succeed in any given instance such that the same instance but with one less manipulator is manipulable. 
Thus, it can be viewed as a $1$-additive approximation algorithm (this approximate sense was introduced by Zuckerman et al. \cite{zuckerman2009algorithms} when analyzing Borda as a social choice function (SCF)).

In order to prove Theorem~\ref{thm:coalition_manip_man} we use the following definitions.
Let $D_0 = \set{d_0}$, where $(d_0, \calL) > (w, \calL)$, and $d_0,w \in B$.
For each $s=1,2,...$, let $D_s \subseteq B$ be $D_s=D_{s-1} \cup \set{d \in B: d$ was ranked above some $d' \in D_{s-1}$ according to ${\calF}(\calL, \calL_R)$ in some stage $l$, $1 \leq l \leq |R|}$. Now, let $D= \bigcup _{0 \leq s}D_s$. 
Note that $\forall s$ $D_s \neq D_{s-1}$, and $s$ does not necessarily equal $|R|$.
Let $s_\ell(w)$ be the score of woman $w$ in ${\calF}(\calL, \calL_R)$ after stage $\ell$.

The proof of Theorem~\ref{thm:coalition_manip_man} relies on Lemmata~\ref{lemma:D_low}-\ref{lemma:average_fromMax}, and its general intuition is as follows. Consider the women in $D$: we show that if there exists a manipulation, then Algorithm~\ref{alg:coalition_manip_man} is able to determine the votes in $\calL_R$ such that the average score of the women in $D$ is lower than the score of woman $w^*$. Moreover, a successful manipulation requires that $w^*$ will be ranked higher than any women in $D$, and thus the algorithm may use one additional manipulator.

We begin with a basic lemma that clarifies where Algorithm~\ref{alg:coalition_manip_man} places the women of $D$. 
\begin{lemma} \label{lemma:D_low}
The women in $D$ are placed in each stage $l$, $1 \leq l \leq |R|$ in the $|D|$ lowest positions.
\end{lemma}

We now show the relation between the score of $w^*$ and the average score of the women in $D$, when there are $|R|-1$ manipulators. In essence, the lemma characterizes the settings in which Algorithm~\ref{alg:coalition_manip_man} returns false and no manipulation exists.   

\begin{lemma} \label{lemma:times_algorithm_returns_false}
Let $q(D)$ be the average score of women in $D$ after $|R|-1$ stages. That is, $q(D) = {\frac{1}{|D|}}\sum_{d \in D}s_{|R|-1}(d)$.
If $s_{|R|-1}(w^*) < q(D)$ and there are $|R|-1$ manipulators, then there is no manipulation that makes $w^*$ the match of $\hat{m}$, and the algorithm will return false.
\end{lemma}

The following lemma shows the relation between the score of $w^*$ and the maximum score of a woman in $D$, when there are $|R|$ manipulators. In essence, the lemma characterizes the settings in which Algorithm~\ref{alg:coalition_manip_man} finds a successful manipulation.

\begin{lemma} \label{lemma:times_algorithm_succeeds}
Assume that a preference order for $\hat{m}$ exists, which makes $w^*$ the match of $\hat{m}$. If $\max_{d \in D}\set{s_{|R|}(d)} < s_{|R|}(w^*)$ and there are $|R|$ manipulators, then there is a manipulation that makes $w^*$ the match of $\hat{m}$, and Algorithm~\ref{alg:coalition_manip_man} will find such a manipulation.
\end{lemma}

Finally, we show that the highest score of a woman from $D$ is not much higher than the average score of the women in $D$. We thus first show that the scores of the women in $D$ are dense, as captured by the following definition.

\begin{definition}[due to~\cite{zuckerman2009algorithms}]
\label{def:1_dense}
A finite non-empty set of integers $A$ is called $1$-dense if, when sorting the set in a non-increasing order $a_1 \geq a_2 \geq \dots \geq a_i$ (such that $ \set{a_1, \dots, a_i}=A$),  $\forall j, 1 \leq j \leq i-1$, $a_{j+1} \geq a_j - 1$ holds.
\end{definition}
\begin{lemma} \label{lemma:1_dense}
Let $D$ be as before. Then the set $\set{s_{|R|-1}(d):d \in D}$ is $1$-dense.
\end{lemma}

\begin{lemma} \label{lemma:average_fromMax}
$\max_{d \in D}\set{s_{|R|}(d)} \leq q(D) + |D| - 1$.
\end{lemma}

Now we can prove the theorem.
\begin{proof} [Proof of Theorem~\ref{thm:coalition_manip_man}]
Clearly, if Algorithm~\ref{alg:coalition_manip_man} returns a preference profile $\calL_R$, then it is a successful manipulation that will make $w^*$ the match of $\hat{m}$.
Suppose that a preference profile exists that makes $w^*$ the match of $\hat{m}$ with $|R|-1$ manipulators. 
By Lemma~\ref{lemma:average_fromMax},
$\max_{d \in D}\set{s_{|R|}(d)} \leq q(D) + |D| - 1$.
By Lemma~\ref{lemma:times_algorithm_returns_false}, 
$
q(D) + |D| - 1 \leq s_{|R|-1}(w^*) + |D| - 1.
$
Since $|D| \leq k - 1$,
$
s_{|R|-1}(w^*) + |D| - 1 < s_{|R|-1}(w^*) + k - 1 = s_{|R|}(w^*).
$
Overall, 
$
\max_{d \in D}\set{s_{|R|}(d)} < s_{|R|}(w^*),
$
and by Lemma~\ref{lemma:times_algorithm_succeeds} the algorithm will find a preference profile that will make $w^*$  the match of $\hat{m}$ with $|R|$ manipulators.
\end{proof}

\section{Women's Side}
We now consider the second variant, in which a specific voter, or a coalition of voters, are associated with an agent $\hat{w}$, and they would like to manipulate the election so that a preferred spouse $m^*$ will be the match of $\hat{w}$.
This variant is more involved, since manipulation of the GS algorithm is also possible by a single woman or a coalition of women. Indeed, there are notable differences between manipulation from the women's side and manipulation from the men's side. First, 
the manipulators from the women's side need to ensure that \textbf{two} men are positioned ``relatively'' high. In addition, the set $B$, which is the set of agents that are placed in low positions, is defined differently, and it is not built iteratively. Finally, in manipulation from the women's side, it is not always possible to place all the agents from $B$ in the lowest positions.

\subsection{Single Manipulator}
With a single manipulator, the Manipulation in the context of Matching from the Women's side (MnM-w) is defined as follows:

\begin{definition}[MnM-w]
    We are given a woman $\hat{w}$, the preference profile $\calL$ of the honest voters that associate with $\hat{w}$, the preference profile $\succ_M$, the preference profile $\succ_{W-\hat{w}}$, a specific manipulator $r$, and a preferred man $m^* \in M$. We are asked whether a preference order $\ell_r$ exists such that $\mu(\hat{w}, \ell_r) = m^*$.
\end{definition}
\begin{algorithm}[bthp]
\caption{Manipulation by a single voter from the women's side}
\label{alg:woman_single_manip}
\SetAlgoLined
\LinesNumbered
\SetAlgoNoEnd
\For{\textbf{each } $m_{nd} \in M \setminus \set{m^*}$}{\label{line:begin_identify_m_nd}
    \tcp{phase 1:}
    $\ell_r \leftarrow$ empty preference order \\
    place $m_{nd}$ in the highest position in $\ell_r$ \\
    place $m^*$ in the second-highest position in $\ell_r$ \\
    \If{$(m_{nd},\calL, \ell_r) > (m^*,\calL, \ell_r) $}{ 
        place $m^*$ in the highest position in $\ell_r$ \\ 
        place $m_{nd}$ in  $\ell_r$ in the highest position such that $(m^*,\calL, \ell_r) > (m_{nd},\calL, \ell_r)$, if such position exists \\
        \If{no such position exists}{
            \textbf{continue} to the next iteration
        }
    }
    \If{$\mu(\hat{w},\ell_r) \neq m^*$ or $m_{nd} \notin o(\hat{w},\ell_r)$ }{
        \textbf{continue} to the next iteration
    }
\tcp{phase 2:}
\For{\textbf{each } $m \notin o(\hat{w}, \ell_r)$}{
    place $m$ in the highest available position in $\ell_r$ \\
}
\tcp{phase 3:}
$B^{nd} \leftarrow o(\hat{w}, \ell_r) \setminus \set{m^*, m_{nd}}$ \\ \label{line:woman_begin_B}
\While{$B^{nd} \neq \emptyset$}{ 
    $b \leftarrow$ the least preferred man from $B^{nd}$ according to $\calF(\calL)$ \\
    place $b$ in the highest available position in $\ell_r$ \\
    remove $b$ from $B^{nd}$ \\
} \label{line:woman_end_B}
\If{$\mu( \hat{w}, \ell_r) = m^*$}{
    \textbf{return} $\ell_r$ \label{line:algtwo_return_ellr}
}
} \label{line:end_identify_m_nd}
\textbf{return} false 
\end{algorithm}
Clearly, if $\mu(\hat{w}) = m^*$ then finding a preference order $\ell_r$ such that  $\mu(\hat{w}, \ell_r) = m^*$ is trivial. We thus henceforth assume that $\mu(\hat{w}) \neq m^*$.
The MnM-w problem can be decided in polynomial-time, using  Algorithm~\ref{alg:woman_single_manip}. The algorithm tries to identify a man $m_{nd} \in M$, and to place him and $m^*$ in $\ell_r$ such that $m_{nd}$ is ranked in $\calF(\calL \cup \set{\ell_r})$ as high as possible while $m^*$ is still preferred over $m_{nd}$ according to $\calF(\calL \cup \set{\ell_r})$. In addition, the algorithm ensures (at the end of phase $1$) that $\mu(\hat{w},\ell_r)=m^*$ and $m_{nd} \in o(\hat{w},\ell_r)$. Note that we compute $\calF(\calL \cup \set{\ell_r})$ even though $\ell_r$ is not a complete preference order, since we assume that all the men that are not in $\ell_r$ get a score of $0$ from $\ell_r$. 
%
%
If phase $1$ is successful (i.e., $\mu(\hat{w},\ell_r)=m^*$ and $m_{nd} \in o(\hat{w},\ell_r)$), the algorithm proceeds to phase $2$, where it fills the preference order $\ell_r$ by placing all the men that are not in $o(\hat{w}, \ell_r)$ in the highest available positions. Finally, in phase $3$, the algorithm places all the men from $o(\hat{w}, \ell_r)$ (except for $m^*$ and $m_{nd}$ that are already placed in $\ell_r$) in the lowest positions in $\ell_r$, and they are placed in a reverse order with regard to their order in $\calF(\calL)$. If $\mu(\hat{w},\ell_r)=m^*$ then we are done; otherwise, the algorithm iterates and considers another man. 

For proving the correctness of Algorithm~\ref{alg:woman_single_manip} we need the following  result.
\begin{lemma}[Swapping lemma, due to~\cite{vaish2017manipulating}] \label{lemma:swapping}
    Given a woman $w \in W$, let $\succ'_w$ be a preference order that is derived from $\succ_w$ by swapping the positions of an adjacent pair of men $(m_i, m_j)$ and making no other changes.
    Then,
    \begin{enumerate}
        \item if $m_i \notin o(w)$ or $m_j \notin o(w)$, then $\mu_x(w,\succ'_w) = \mu(w)$. \label{case:not_proposal}
        \item if both $m_i$ and $m_j$ are not one of the two most preferred proposals among $o(w)$ according to $\succ_w$, then $\mu_x(w,\succ'_w) = \mu(w)$. \label{case:not_2_preferred_proposals}
        \item if $m_i$ is the second preferred proposal among $o(w)$ according to $\succ_w$ and $m_j$ is the third preferred proposal among $o(w)$ according to $\succ_w$, then $\mu_x(w,\succ'_w)$ $\in \set{\mu(w), m_j}$. \label{case:2_3_preferred_proposals}
        \item if $m_i = \mu(w)$ and $m_j$ is the second preferred proposal among $o(w)$ according to $\succ_w$, then the second preferred proposal among $o(w)$ according to $\succ_w'$ is $m_i$ or $m_j$. \label{case:2_preferred_proposals}
    \end{enumerate}
\end{lemma}
If we use the swapping lemma sequentially, we get the following corollary.
\begin{corollary} \label{corollary:m_nd}
Given a woman $w \in W$, let $\succ'_w$ be a preference order for $w$ such that $\succ_w \neq \succ_w'$.
Let $m^* \in M$ be the most preferred man among $o(w)$ according to $\succ_w$. That is, $\mu(w) = m^*$.
Let $m_{nd} \in M$ be the second most preferred man among $o(w)$ according to $\succ_w$. 
If $m_{nd}$ is the most preferred man among $o(w) \setminus \set{m^*}$ according to $\succ_w'$, and $m^* \succ_w' m_{nd}$, then $o(w) = o_x(w,\succ'_w)$ and thus $\mu_x(w,\succ'_w)$ $= \mu(w) = m^*$.
\end{corollary}
Corollary~\ref{corollary:m_nd} is the basis of our algorithm. Intuitively, the manipulator needs to ensure that $m^*$ is among the set of proposals $o(\hat{w},\ell_r)$, and that $m^*$ is the most preferred men, according to $\calF(\calL \cup \set{\ell_r})$, among this set.
That is, $m^*=\mu(\hat{w},\ell_r)$. Thus, the algorithm searches for a man, denoted by $m_{nd}$, that serves as the second-best proposal. If such a man exists, then, according to Corollary~\ref{corollary:m_nd}, the position of every man $m \in o(\hat{w},\ell_r)$ does not change $\hat{w}$'s match (which is currently $m^*$) if $m_{nd}$ is preferred over $m$ in $\calF(\calL \cup \set{\ell_r})$. In addition, the position of every man $m \notin o(\hat{w},\ell_r)$ does not change $\hat{w}$'s match at all. 
\begin{theorem}
\label{thm:alg_2_correct}
Algorithm ~\ref{alg:woman_single_manip} correctly decides the MnM-w problem in polynomial time.
\end{theorem}

\subsection{Coalitional Manipulation} \label{sec:calition_manipulation}

Finally, We study manipulation by a coalition of voters from the women's side. 

\begin{definition}[coalitional MnM-w] \label{def:coalition_among_a_woman}
    We are given a woman $\hat{w}$, the preference profile $\calL$ of the honest voters that associate with $\hat{w}$, the preference profile $\succ_M$, the preference profile $\succ_{W-\hat{w}}$, a coalition of manipulators $R$, and a preferred man $m^* \in M$. We are asked whether a preference profile $\calL_R$ exists such that $\mu(\hat{w}, \calL_R) = m^*$.
\end{definition}

Similar to the single manipulator setting, if $\mu(\hat{w}) = m^*$ then finding a preference profile $\calL_R$ such that  $\mu(\hat{w}, \ell_R) = m^*$ is trivial. We thus henceforth assume that $\mu(\hat{w}) \neq m^*$.
The coalitional MnM-w problem is computationally hard, even with two manipulators, and we again reduce from the Permutation Sum problem (Definition~\ref{def:permutationSum}).

\begin{theorem} \label{thm:coalitional_mnm-w}
Coalitional MnM-w is NP-Complete.
\end{theorem}

Similar to the coalitional MnM-m, the coalitional MnM-w also has an efficient heuristic algorithm that finds a successful manipulation.
We use Algorithm ~\ref{alg:coalition_manip_woman}, which works as follows.
\begin{algorithm}[hbpt]
\caption{Manipulation by a coalition of voters from the women's side}
\label{alg:coalition_manip_woman}
\SetAlgoLined
\LinesNumbered
\SetAlgoNoEnd
    \For{\textbf{each } $m_{nd} \in M \setminus \set{m^*}$}{
        \tcp{phase 1:}
        $gap \leftarrow s(m_{nd}, \calL) - s(m^*, \calL)$ \\
        \If{$m_{nd}$ is preferred over $m^*$ according to the lexicographical tie breaking rule}{ 
            $gap = gap + 1$ \\
        }
        \If{$|R| \cdot (k-1) < gap$}{ \label{line:gap_too_big}
            \textbf{continue} to the next iteration 
        } 
        $\calL_R \leftarrow (\ell_1,..., \ell_{|R|})$ where each preference order is an empty one \\
        \If{$|R| \geq gap$}{ \label{line:begin_2_preferred}
            place $m^*$ in the highest position and $m_{nd}$ in the second highest position, in $\max(gap + \lceil (|R| - gap) / 2  \rceil,0)$ preference orders of $\calL_R$ \\ \label{line:first_place_2_preferred}
            place $m^*$ in the second highest position and $m_{nd}$ in the highest position in all of the other preference orders of $\calL_R$ \\ \label{line:end_2_preferred}
        }
        \Else{ 
            place $m^*$ in the highest position in each $\ell_r \in \calL_R$ \\
            $s_{m_{nd}} \leftarrow |R|\cdot (k-1) - gap$ \\
            place $m_{nd}$ in $(s_{m_{nd}} \mod |R|)$ manipulators such that it gets a score of $\lceil \frac{s_{m_{nd}}}{|R|} \rceil$ from each manipulator  \\ 
            place $m_{nd}$ in the other manipulators such that it gets a score of $\lfloor \frac{s_{m_{nd}}}{|R|} \rfloor$ from each manipulator  \\ 
        }
        \If{$\mu(\hat{w}, \calL_R) \neq m^*$ or $m_{nd} \notin o(\hat{w},\calL_R)$}{ \label{line:pass_phase_1}
            \textbf{continue} to the next iteration
        }
        \tcp{phase 2:}
        $B^{nd} \leftarrow o(\hat{w}, \calL_R) \setminus \set{m^*, m_{nd}}$ \\ \label{line:Bnd_definition}
        \For{\textbf{each } $r \in R$} { \label{line:coalition_men_begin_stage}
            \For{\textbf{each } $m \in M \setminus (B^{nd} \cup \set{m^*, m_{nd}})$}{
                place $m$ in the next highest available position in $\ell_r$ \\
            }
            $B' \leftarrow B^{nd}$ \\
            \While{$B' \neq \emptyset$} {
                $b \leftarrow$ the least preferred man from $B'$ according to $\calF(\calL \cup \calL_R)$ \\
             place $b$ at the highest available position in $\ell_r$ \\
             remove $b$ from $B'$ \\
            }
        } \label{line:coalition_men_end_stage}
        \If{$\mu(\hat{w}, \calL_R) = m^*$}{
            \textbf{return} $\calL_R$    
        }
    }
    \textbf{return} false
\end{algorithm}
Similar to Algorithm \ref{alg:woman_single_manip}, Algorithm \ref{alg:coalition_manip_woman} needs to identify a man $m_{nd} \in M$, such that $m_{nd}$ is ranked in $\calF(\calL \cup \calL_R)$ as high as possible while $m^*$ is still preferred over $m_{nd}$ according to $\calF(\calL \cup \calL_R)$.
In addition, the algorithm needs to ensure that $\mu(\hat{w},\calL_R)=m^*$ and $m_{nd} \in o(\hat{w},\calL_R)$, which is done at the end of phase $1$.
Indeed, finding such a man $m_{nd} \in M$, and placing him and $m^*$ in every $\ell_r \in \calL_R$ is not trivial. 
The algorithm considers every $m \in M \setminus \set{m^*}$, and computes the difference between the score of $m$ from $\calL$ and the score of $m^*$ from $\calL$. 
Clearly, if this gap is too big, $m$ cannot be $m_{nd}$ (line~\ref{line:gap_too_big}). Otherwise, there are two possible cases.
If there are many manipulators, specifically, $|R| \geq gap$, then the algorithm places $m^*$ and $m$ in the two highest positions in every $\ell_r$ (lines~\ref{line:first_place_2_preferred}-\ref{line:end_2_preferred}). 
On the other hand, if $|R| < gap$, then the algorithm places $m^*$ in the highest position in every $\ell_r$. 
The algorithm places $m$ such that he gets a total score of $|R|\cdot (k-1)-gap$ from the manipulators. Moreover, the algorithm tries to place $m$ in almost the same position in every $\ell_r$.
If phase $1$ is successful, the algorithm proceeds to fill the preference orders of $\calL_R$ iteratively in phase $2$.
The algorithm first defines the set $B^{nd}$, which consists of all the men from $o(\hat{w}, \calL_R)$, except for $m^*$ and $m_{nd}$. Note that at the beginning of phase $2$, in which $B^{nd}$ is defined, only $m^*$ and $m$ are positioned in every $\ell_r \in \calL_R$.
Then, in every $\ell_r \in \calL_R$, the algorithm places all the men that are not in $B^{nd}$ (except for $m^*$ and $m_{nd}$ that are already placed in $\calL_R$) in the highest available positions. The algorithm places all the men from $B^{nd}$ in the lowest positions in $\ell_r$, and they are placed in a reverse order with regard to their current order in $\calF(\calL \cup \calL_R)$. Note that since $\calL_R$ is updated in every iteration, the order in which the men from $B^{nd}$ are placed in each $\ell_r$ is not the same for each $r$. If $\mu(\hat{w},\ell_R)=m^*$ then we are done; otherwise, the algorithm iterates and considers another man. We refer to each of the iterations in Lines \ref{line:coalition_men_begin_stage}-\ref{line:coalition_men_end_stage} as a \emph{stage} of the algorithm.

We now show that Algorithm~\ref{alg:coalition_manip_woman} will succeed in any given instance such that the same instance but with two less manipulators is manipulable. That is, unlike the coalitional MnM-m, the coalitional MnM-w admits a $2$-additive approximation algorithm.
Formally,
\begin{theorem} \label{thm:coalition_manip_woman}
Given an instance of coalitional MnM-w,
    \begin{enumerate}
        \item If there is no preference profile making $m^*$ the match of $\hat{w}$ exists, then Algorithm~\ref{alg:coalition_manip_woman} will return false. 
        \item If a preference profile making $m^*$ the match of $\hat{w}$ exists, then for the same instance with two additional manipulators, Algorithm~\ref{alg:coalition_manip_woman} will return a preference profile that makes $m^*$ the match of $\hat{w}$.
    \end{enumerate}
\end{theorem}

\section{Conclusion and Future Work}
In this paper, we initiate the analysis of strategic voting in the context of stable matching of teams. Specifically, we assume that the Borda rule is used as a SWF, which outputs an order over the agents that is used as an input in the GS algorithm. Note that in the standard model of manipulation of Borda, the goal is that a specific candidate will be the winner. In our setting, the algorithms need also to ensure that specific candidates will not be ranked too high. Similarly, in the standard model of manipulation of the GS algorithm, the goal is simply to achieve a more preferred match. In our setting, the algorithms for manipulation need also to ensure that a less preferred spouse is matched to a specific agent. Therefore, even though the manipulation of the Borda rule and the manipulation of the GS algorithm have already been studied, our analysis of the manipulation of Borda rule in the context of GS stable matching provides a better understanding of both algorithms.

Interestingly, our algorithms for the single manipulator settings are quite powerful. They provide exact solutions for the single manipulator case, and their generalizations provide approximate solutions to the coalitional manipulation settings, both when the manipulators are on the men’s side or on the women’s side.
%
%

For future work, we would like to extend our analysis and study additional voting rules as SWFs. It is also worth studying the destructive manipulation objective. In addition, it will be interesting to examine a new type of manipulation, which is only relevant in our setting, in which there is a coalition of manipulators, but every manipulator is associated with a different agent.

\section*{Acknowledgments}
This research has been partly supported by the Israel Science Foundation under grant 1958/20, by the EU Project TAILOR under grant 952215, and by the Ministry of Science, Technology \& Space (MOST), Israel.

\bibliographystyle{plainnat}
\bibliography{sample-bibliography}

\begin{contact}
Leora Schmerler\\
Ariel University\\
Ariel, Israel\\
\email{leoras@ariel.ac.il}
\end{contact}

\begin{contact}
Noam Hazon\\
Ariel University\\
Ariel, Israel\\
\email{noamh@ariel.ac.il}
\end{contact}

\begin{contact}
Sarit Kraus\\
Bar-Ilan University\\
Rammat-Gan, Israel\\
\email{sarit@cs.biu.ac.il}
\end{contact}

\clearpage
\appendix
\section{Table of Notations}

\begin{table*}[h!]
\begin{center}
 \begin{tabular}{||c c||}
 \hline
 $M / W$ & The set of men / women \\
 \hline  
 $k$ & $|M|$ and $|W|$ \\
 \hline
 $\succ_m$ / $\succ_w$ & The linear preference order of man $m$ / woman $w$ \\
  \hline
 $\succ_M$ / $\succ_W$ & The set of preference orders of the men / women \\
 \hline
 $\succ_{M-m}$ / $\succ_{W-w}$ & The set of preference orders of the men / women without $m$ / $w$ \\
 \hline
 $\mu(m)$ / $\mu(w)$ & The match $m$ / $w$ received \\
 \hline
 
 $\hat{m}$ / $\hat{w}$ & A man / woman that represents a voting group  \\
 \hline
 $V$ & The set of honest voters associated with $\hat{w}$ or $\hat{m}$ \\
 \hline
 $\calL$ & The set of preference order of $V$ \\
 \hline
 $\calF(\calL)$ & The output of the SWF on the preferences set $\calL$  \\
 \hline

 $w^*$ / $m^*$ & The preferred woman / man according to the preference of $r$ \\
 \hline
 
 $\mu(\hat{w})$ / $\mu(\hat{m})$ & The match $\hat{w}$ / $\hat{m}$ received according to $\hat{w}$'s / $\hat{w}$'s preference order $\calF(\calL)$\\
 \hline
 $\mu(\hat{w}, \ell_r)$ / $\mu(\hat{m}, \ell_r)$ & The match $\hat{w}$ / $\hat{m}$ received according to $\hat{w}$'s / $\hat{m}$'s preference order $\calF(\calL \cup \set{\ell_r})$\\
 \hline
 $o(\hat{w})$ & The set of proposals $\hat{w}$ received according to $\hat{w}$'s preference order $\calF(\calL)$\\
 \hline
 $o(\hat{w}, \ell_r)$ & The set of proposals $\hat{w}$ received according to $\hat{w}$'s preference order $\calF(\calL \cup \set{\ell_r})$\\
 \hline
 

 $r$ & The manipulator \\
 \hline
 $\ell_r$ & The linear preference order of $r$ \\ 
 \hline
 $R$ & The set of manipulators \\
 \hline
 $\calL_R = \set{\ell_{1}, ..., \ell_{|R|}}$ & The preference profile of $R$ \\
 \hline

 $\calF(\calL \cup \set{\ell_r})$ & The output of the SWF on the preferences set $\calL \cup {\ell_r}$  \\
 \hline
 $s(c,\ell_v)$ & The score of $c$ from the preferences of $v$ according to the voting rule \\
 \hline
 $s(c,\calL)$ & The score of $c$ from the preferences set $\calL$, i.e., $s(c,\calL) = \sum_{v \in V} s(c,v)$ \\
 \hline
  $s(c,\calL, \ell_r)$ & The score of $c$ from the preferences set $\calL \cup \set{\ell_r}$ \\
 \hline
 $(c,\ell) > (c',\ell')$ & $s(c,\ell) > s(c',\ell')$ \\ & or $s(c,\ell) = s(c',\ell')$ but $c$ is preferred over $c'$ \\ & according to a lexicographical tie-breaking rule \\
 \hline
 $(c,\calL, \ell) > (c',\calL, \ell')$ & $s(c,\calL, \ell) > s(c',\calL, \ell')$ \\ & or $s(c,\calL, \ell) = s(c',\calL, \ell')$ but $c$ is preferred over $c'$ \\ & according to a lexicographical tie-breaking rule \\
 \hline

 $c >_\ell c'$ & $(c,\ell) \succ (c',\ell)$ \\
 \hline

\end{tabular}
\end{center}
    \caption{Table of Notations.}
\end{table*}

\section{Omitted Proofs}
\subsection{Proof of Corollary~\ref{corollary:placing_higher}}

\begin{proof}
Assume by contradiction that $w_m' = \mu_x(m,\succ'_m) \neq w_m$. There are two possible cases:
\begin{enumerate}
    \item $w_m' \succ_m w_m$.
    That is, $m$ will have an incentive to strategically report the preference order $\succ'_m$ instead of his sincere preference order $\succ_m$, in contradiction to Theorem~\ref{thm:truthful}.
    \item $w_m \succ_m w_m'$.
    By the definition of $\succ'_m$, it is not possible that $w_m' \succ'_m w_m$.
    Consider the setting in which $\succ_W$ and $\succ_{M-m}$ are the same, but the sincere preference order of $m$ is $\succ'_m$. However, $w_m \succ'_m w_m'$, and thus, in this setting, $m$ will have an incentive to strategically report the preference order $\succ_m$ instead of his sincere preference order $\succ'_m$, in contradiction to Theorem~\ref{thm:truthful}.
\end{enumerate}
\end{proof}

\subsection{Proof of Lemma~\ref{lemma:exists_pref_at_all}}
%
\begin{proof}
One direction is trivial, so we prove the other direction.
Assume that a preference order $\succ_t$ for $\hat{m}$ exists such that $w^*=\mu_x(\hat{m}, \succ_t)$. We show  that $w^*=\mu_x(\hat{m}, \succ_x)$.
Assume to contradiction that $\mu_x(\hat{m}, \succ_x) = w$ and $w \neq w^*$. It is not possible that $w \succ_x w^*$, since the condition in line \ref{line:begin_single_find_B} would have been true (where $b=w$), and then the algorithm would have placed $w$ below $w^*$ in $\succ_x$ (in line \ref{line:move_b}). Now, consider the setting in which $\succ_W$ and $\succ_{M-\hat{m}}$ are the same, but the sincere preference order of $\hat{m}$ is $\succ_x$. In this setting, if $\hat{m}$ is truthful, then his match is $w$. However, if $\hat{m}$ strategically reports the preference order $\succ_t$, then his match will be $w^*$ and $w^* \succ_x w$, in contradiction to Theorem~\ref{thm:truthful}.
\end{proof}

\subsection{Proof of Lemma~\ref{lemma:b_prevent_m*}}

\begin{proof}
Let $b \in B$ such that $b \succ_t w^*$.
Since $b \in B$, then there is an iteration of Algorithm~\ref{alg:man_single_manip} in which $b$ is moved immediately below $w^*$ in $\succ_x$ (Line~\ref{line:move_b}).
Let $\succ^{(1)}_x$ be the preference order $\succ_x$ in this iteration, before the execution of Line~\ref{line:move_b}.
By definition, $\mu_x(\hat{m}, \succ^{(1)}_x) = b$.
Let $\succ^{(2)}_x$ be $\succ^{(1)}_x$ when we move all the woman $w$ such that $b \succ_t w \succ_t w^*$ to be below $b$ and above $w^*$ in $\succ^{(2)}_x$, and  all the woman $w$ such that $w^* \succ_t w$ to be below $w^*$ in $\succ^{(2)}_x$. 
According to Corollary~\ref{corollary:placing_higher}, $\mu_x(\hat{m}, \succ^{(2)}_x) = b$.
Let $\succ^{(3)}_x$ be $\succ^{(2)}_x$ when we move all the woman $w$ such that $w \succ_t b$ to be above $b$ in $\succ^{(3)}_x$.
Assume by contradiction that $\mu_x(\hat{m}, \succ_t)$ $= w^*$. By Lemma~\ref{lemma:permutation}, $\mu_x(\hat{m}, \succ^{(3)}_x) = w^*$. On the other hand, each woman $w$ such that $w^* \succ^{(3)}_x w$ is also below $w^*$ in $\succ^{(2)}_x$. Therefore, according to Corollary~\ref{corollary:placing_higher}, $\mu_x(\hat{m}, \succ^{(2)}_x) = w^*$, which is a contradiction. Therefore, $\mu_x(\hat{m}, \succ_t) \neq w^*$.
\end{proof}

\subsection{Proof of Theorem~\ref{thm:alg_1_correct}}
\begin{proof}
Clearly, the algorithm runs in polynomial time since there are three loops, where the three loops together iterate at most $2k$ times, and the running time of the GS matching algorithm is in  $O(k^2)$. In addition, if the algorithm returns a preference order, which is a manipulative vote for the manipulator $r$, then $w^*$ will be the match of $\hat{m}$ by the GS algorithm. We need to show that if there exists a preference order for the manipulator $r$ that makes $w^*$ the match of $\hat{m}$, then our algorithm will find such a preference order for $r$.
Assume that a manipulative vote, $\ell_t$, exists, which makes $w^*$ the match of $\hat{m}$. That is, $\mu(\hat{m}, \ell_t) = w^*$. Then, by Lemma~\ref{lemma:exists_pref_at_all}, the algorithm finds a preference order for $\hat{m}$ that makes $w^*$ his match (i.e., the preference order $\succ_x$), and thus it does not return false in line~\ref{line:single_man_no_preference}.
We show that Algorithm~\ref{alg:man_single_manip} returns $\ell_r$ in line~\ref{line:single_man_success}.
we begin by analyzing $\ell_r$ after phase $1$. Since $w^*$ is positioned in the highest position in $\ell_r$, then $s(w^*, \ell_r) \geq s(w^*, \ell_t)$. Since $\ell_t$ is a successful manipulation, then $w^*$ is preferred over every $b \in B$ in $\calF(\calL \cup \set{\ell_t})$ according to Lemma~\ref{lemma:b_prevent_m*}. In addition, every $b \in B$ gets a score of $0$ from $\ell_r$, and thus $w^*$ is preferred over every $b \in B$ in $\calF(\calL \cup \set{\ell_r})$. Note that when Algorithm~\ref{alg:man_single_manip} changes $\calF(\calL)$ to $\succ_x$ it only re-position the women $b \in B$. That is, for every $w \in W \setminus B$ such that $w^* \succ_x w$, $w$ is also less preferred than $w^*$ in $\calF(\calL)$. However, since $w^*$ is positioned in the highest position in $\ell_r$, then for every $w \in W \setminus B$ such that $w^* \succ_x w$, $w$ is also less preferred than $w^*$ in $\calF(\calL \cup \set{\ell_r})$.
Overall, every $w$ such that $w^* \succ_x w$ is also less preferred than $w^*$ in $\calF(\calL \cup \set{\ell_r})$, and thus according to Corollary~\ref{corollary:placing_higher}, $\mu(\hat{m},\ell_r) = w^*$.

We now show that Algorithm~\ref{alg:man_single_manip}  (lines~\ref{line:begin_single_man_reverse}-\ref{line:end_single_man_reverse}) can assign scores to all the women $b \in B$ such that $\mu(\hat{m},\ell_r) = w^*$ in line~\ref{line:single_man_success} (i.e., $\ell_r$ is a successful manipulation). 
For any $b \in B$, if $s(b,\ell_r) \leq s(b,\ell_t)$ then $s(b,\calL,\ell_r) \leq s(b,\calL,\ell_t)$. Since $(w^*,\calL,\ell_t) > (b,\calL,\ell_t)$ then $(w^*,\calL,\ell_r) > (b,\calL,\ell_r)$. 
Otherwise, let $b \in B$ be a woman such that $s(b,\ell_r) > s(b,\ell_t)$ and let $s = s(b, \ell_r)$. 
There are $s$ women from $B$ below $b$ in $\ell_r$. 
According to the pigeonhole principle, at least one of the women from $B$, denoted $b'$, gets a score of at least $s$ from $\ell_t$.
That is, $s(b', \ell_t) \geq s(b, \ell_r)$. By the algorithm construction, all of the women $b'' \in B$ that are positioned lower than $b$ in $\ell_r$ are positioned higher than $b$ in ${\calF}(\calL)$. That is, $(b',\calL) > (b,\calL)$. 
However, $(b',\calL, \ell_t) < (w^*,\calL, \ell_t)$ and thus $(b,\calL, \ell_r) < (w^*,\calL, \ell_r)$. 
Overall, after placing the women from $B$ in $\ell_r$, $\forall b \in B$, $(b,\calL, \ell_r) < (w^*,\calL, \ell_r)$.
That is, every $w$ such that $w^* \succ_x w$ is also less preferred than $w^*$ in $\calF(\calL \cup \set{\ell_r})$, and thus according to Corollary~\ref{corollary:placing_higher}, $\mu(\hat{m},\ell_r) = w^*$.
\end{proof}

\subsection{Proof of Theorem~\ref{thm:coalitional_mnm-m}}

\begin{proof}
Clearly, the coalitional MnM-m problem is in $NP$. 
Given an instance of the Permutation Sum problem, we build an instance of the coalitional MnM-m problem as follows. There are $q + 3$ men and $q + 3$ women. 
The men are: $m_1, ..., m_{q+2}$ and $\hat{m}$. 
The women are: $w_1, ..., w_{q+2}$ and $w^*$. We assume that $w^*$ is preferred over the other women according to the lexicographical tie-breaking rule.
The $q$ women, $w_1,\ldots,w_q$, correspond to the integers $X_1,\ldots,X_q$. 
By Lemma $1$ from~\cite{davies2011complexity}, we can construct an election in which the non-manipulators cast votes such that:
\[
\hat{m} = (w_{q+1}, w_1, \ldots, w_q, w^* ,w_{q + 2}),
\]
 and the corresponding scores are: 
\[
(2q + 4 + C, 2q + 4 + C - X_1 \ldots, 2q + 4 + C - X_q, C, z),
\] 
where $C$ is a constant and $z \leq C$.
The preference orders of the men beside $\hat{m}$ are the same, and are as follows:
\[
(w^*,w_1,\ldots,w_{q + 2})
\] 
The preference orders of the women are the same, and are as follows:
\[
(\hat{m},m_1,\ldots,m_{q + 2})
\] 
We show that two manipulators can make the man $\hat{m}$ the match of woman $w^*$ iff the Permutation Sum problem has a solution.

($\Leftarrow$) Suppose we have two permutations $\sigma$ and $\pi$ of $1$ to $q$ such that $\sigma (i) + \pi (i) = X_i$. Let $\sigma^{-1}$ be the inverse function of $\sigma$, i.e., $i=\sigma^{-1}(w)$ if $w=\sigma(i)$. We define $\pi^{-1}(w)$  similarly.
We construct the following two manipulative votes:
\[
(w^*,w_{q+2},w_{\sigma^{-1}(q)},\ldots,w_{\sigma^{-1}(1)},w_{q+1})
\]
\[
(w^*,w_{q+2}, w_{\pi^{-1}(q)}, \ldots, w_{\pi^{-1}(1)},w_{q+1})
\]
Since $\sigma (i) + \pi (i) = X_i$
and $z \leq C$, 
the preference order of $\hat{m}$, $\calF(\calL_R, \calL)$ is: 
\[
(w^*, w_1, \ldots,w_q, w_{q+1} ,w_{q+2})
\]
since the corresponding scores are:
\[
(2q+4+C,2q+4+C-X_1+X_1, \ldots, 2q+4+C-X_q+X_q, 
\]
\[
2q+4+C, 2(q + 1) + z).
\]
Therefore, $w^*$ is the match of $\hat{m}$.

($\Rightarrow$) Assume we have a successful manipulation. 
Such manipulation must ensure that none of the women $w_1,\ldots ,w_{q+2}$ are placed in the first positions in the preference order of $\hat{m}$.
That is, to ensure that woman $w^*$ is ranked higher than woman $w_{q+1}$, both manipulators have to place $w^*$ in the highest position in their preferences, and $w_{q+1}$ in the lowest position in their preferences. Thus, the score of woman $w^*$ in the preference of $\hat{m}$ will be $2q+4+C$.
Let $\sigma(i)$ be a function that determines the score where the first manipulator is assigned to woman $w_i$. $\pi(i)$ is defined similarly for the second manipulator. 
Since the manipulation is successful, for every $i$, $1 \leq i \leq q$, 
$$
2q+4+C-X_i +\sigma (i)+ \pi (i) \leq 2q+4+C,
$$
and thus, 
$$
\sigma (i)+ \pi (i) \leq X_i.
$$
Since $\sum_{i=1}^q X_i = q(q+1)$,
$$
\sum_{i=1}^q \sigma(i)+\pi(i) \leq q(q+1).
$$
On the other hand, since $w_{q+1}$ is placed in the lowest position by both manipulators, 
$$
\sum_{i=1}^q \sigma(i) \geq \frac{q(q+1)}{2}
$$
and
$$
\sum_{i=1}^q \pi(i) \geq \frac{q(q+1)}{2}.
$$
Therefore, 
$\sum_{i=1}^q \sigma(i)+\pi(i) = q(q+1)$, and $\sum_{i=1}^q \sigma(i) = \sum_{i=1}^q \pi(i) = \frac{q(q+1)}{2}$. That is, $\sigma$ and $\pi$ are permutations of $1$ to $q$. 
Moreover, since there is no slack in the inequalities,
$$
\sigma (i)+ \pi (i) = X_i.
$$
That is, there is a solution to the Permutation Sum problem.
\end{proof}

\subsection{Proof of Lemma~\ref{lemma:D_low}}
\begin{proof}
Zuckerman et al.~\cite{zuckerman2009algorithms} 
define a set of candidates $G_W$, and show that their algorithm places them at the $|G_W|$ lowest positions (Lemma $3.5$ in \cite{zuckerman2009algorithms}). Even though our definition of the set $D$ is slightly different, our algorithm handles the set $D$ in the same way that their algorithm handles the set $G_W$. Thus, the women in $D$ are placed in each stage in the $|D|$ lowest positions.
\end{proof}

\subsection{Proof of Lemma~\ref{lemma:times_algorithm_returns_false}}
\begin{proof}
Assume that there is a successful manipulation $\calL_T$ with $|R|-1$ manipulators.
Recall that $s(w,\calL,\calL_T)$ is the Borda score of a woman $w$ in  $(\calL, \calL_T)$.
Since Algorithm~\ref{alg:coalition_manip_man} places the woman $w^*$ at the highest positions, then:
{\small 
\begin{equation}\label{eq:sw}
s_{|R|-1}(w^*) = (s(w^*, \calL)+ \sum_{j = 1}^{|R|-1} k-1) \geq s(w^*,\calL, \calL_T)
\end{equation}
}

Since Algorithm~\ref{alg:coalition_manip_man} (as proved in Lemma~\ref{lemma:D_low}) places all of the women $d \in D$ at the $|D|$ lowest positions, then:
{\small 
\begin{equation}\label{eq:qd}
q(D) = {\frac{1}{|D|}}(\sum_{d \in D}s(d, \calL)+ \sum_{j = 1}^{|R|-1} \sum_{i = 0}^{|D|-1}i) \leq
\end{equation}
$$
{\frac{1}{|D|}} \sum_{d \in D} s(d,\calL, \calL_T)=:q^T(D)
$$
}
Combining the equations and the assumption that $s_{|R|-1}(w^*) < q(D)$, we get that $q^T(D) > s(w^*, \calL, \calL_T)$. 
Since $q^T(D)$ is an average, there is at least one $d \in D$ such that $s(d, \calL, \calL_T) \geq q^T(D)$. Therefore, $s(d, \calL, \calL_T) > s(w^*, \calL, \calL_T)$. However, since $d$ is positioned above $w^*$, then $\hat{m}$'s match will not be $w^*$, according to Lemma~\ref{lemma:b_prevent_m*}, which is a contradiction to the assumption that $\calL_T$ is a successful manipulation.
Therefore, there is no manipulation that makes $w^*$ the match of $\hat{m}$. Clearly, if Algorithm~\ref{alg:coalition_manip_man} returns $\calL_R$, then it is a successful manipulation. Thus, if there is no manipulation, then the algorithm will return false. 
\end{proof}

\subsection{Proof of Lemma~\ref{lemma:times_algorithm_succeeds}}
\begin{proof}
Since Algorithm~\ref{alg:coalition_manip_man} builds the preference order $\succ_x$ in an identical way to Algorithm~\ref{alg:man_single_manip}, then according to Lemma~\ref{lemma:exists_pref_at_all}, Algorithm~\ref{alg:coalition_manip_man} will find a preference order $\succ_x$ such that $w^* = \mu_x(\hat{m},\succ_x)$.
Recall that for every $w \in W \setminus \set{B}$ such that $w^* \succ_x w$, $w$ is also less preferred than $w^*$ in $\calF(\calL)$. However, since $w^*$ is positioned in the highest position in each $\ell_r \in \calL_R$, then for every $w \in W \setminus \set{B}$ such that $w^* \succ_x w$, $w$ is also less preferred than $w^*$ in $\calF(\calL \cup \calL_R)$.
We now show that for all $b \in B \setminus D$, it holds that $s_{|R|}(b) \leq min_{d \in D}\{s_{|R|}(d)\}$.
Assume by contradiction that there exists $b \in B \setminus D$ and $d \in D$, such that $s_{|R|}(b) > s_{|R|}(d)$. Then, by the definition of $D$, $b \in D$, which contradicts the choice of $b$.
Since $\max_{d \in D}\set{s_{|R|}(d)} < s_{|R|}(w^*)$, then for all $b \in B$, it holds that $s_{|R|}(b) < s_{|R|}(w^*)$. 
Overall, we get that for all $w \in W$ such the $w^* \succ_x w$ $w$ is also less preferred over $w^*$ according to $\calF(\calL \cup \calL_R)$. Therefore, by Corollary~\ref{corollary:placing_higher}, $w^*$ is $\hat{m}$'s match with the preference $\calF(\calL \cup \calL_R)$.
That is, the algorithm finds a manipulation that makes $w^*$ the match of $\hat{m}$.
\end{proof}

\subsection{Proof of Lemma~\ref{lemma:1_dense}}
\begin{proof}
Zuckerman et al.~\cite{zuckerman2009algorithms} 
define a set of candidates $G_W$, and show that the scores of the  candidates in $G_W$ are $1$-dense (Lemma $3.12$ in \cite{zuckerman2009algorithms}). Even though our definition of the set $D$ is slightly different, the set of scores $\set{s_{|R|-1}(d):d \in D}$ is essentially identical to the set of scores of the candidates in $G_W$. Thus, it is $1$-dense. 
\end{proof}

\subsection{Proof of Lemma~\ref{lemma:average_fromMax}}
\begin{proof}
Sort the members of $D$ by their scores after stage $|R|-1$ in an increasing order, i.e., $D = \set{d_1, \ldots , d_{|D|}}$ such that for all $1 \leq t \leq |D| - 1$, $s_{|R|-1}(d_t) \leq s_{|R|-1}(d_{t + 1})$. 
Thus, $d_1 = \argmin_{d \in D}\set{s_{|R|-1}(d)}$. 
Denote for $1 \leq t \leq |D|$, $g_t = s_{|R|-1}(d_t) + |D| - t$, and let $G = \set{g_1, \ldots , g_{|D|}}$.
Note that according to Algorithm~\ref{alg:coalition_manip_man}, for any $t$, $s_{|R|}(d_t) = g_t$. 
Now, since the set $\set{s_{|R|-1}(d):d \in D}$ is $1$-dense (according to Lemma~\ref{lemma:1_dense}) then for any $1 \leq t \leq |D| - 1$, $g_t \geq g_{t+1}$.
Thus, for any $t > 1$, $g_1 \geq g_t$. That is, $g_1 = \max_{d \in D}\set{s_{|R|}(d)} = s_{|R|-1}(d_1) +|D| - 1 = \min_{d \in D}\set{s_{|R|-1}(d)} + |D| - 1$. 
Clearly, $q(D) \geq \min_{d \in D}\set{s_{|R|-1}(d)}$, and thus $\max_{d \in D}\set{s_{|R|}(d)} \leq q(D) + |D| - 1$.
\end{proof}

\subsection{Proof of Corollary~\ref{corollary:m_nd}}

\begin{proof}
We construct the preference order $\succ_w'$ by starting from $\succ_w$ and performing a sequence of swaps of two adjacent men till the resulting preference order is $\succ_w'$.
We show that each swap does not change the set of proposals, by repeatedly invoking Lemma~\ref{lemma:swapping}.

We begin by positioning the most preferred man according to $\succ'_w$, using swaps of two adjacent men. 
That is, if $m_{st}$ is the most preferred man according to $\succ'_w$, we swap pairs $(m,m_{st}), m \in M$, until $m_{st}$ is placed in the first position in $\succ_w'$. We call these swaps the swaps of $m_{st}$.
We then position the second preferred man using his swaps, and so on.
Clearly, this process terminates since the number of men is finite.
Let $\succ_w^{(t)}$ be $\succ_w$ after $t$ swaps. That is, $\succ_w^{(0)}$ is $\succ_w$, $\succ_w^{(1)}$ is $\succ_w$ after one swap, and $\succ_w^{(t+1)}$ is $\succ_w^{(t)}$ after one swap. We show that for every $t \geq 0$, $o_x(w, \succ_w^{(t)}) = o_x(w, \succ_w^{(t+1)})$ and thus $o(w) = o_x(w, \succ_w^{(0)}) = o_x(w, \succ'_w)$. Let $(m_i,m_j)$ be the pair of adjacent men that swap their positions when moving from $\succ_w^{(t)}$ to $\succ_w^{(t+1)}$. That is, $m_i \succ_w^{(t)} m_j$ and $m_j \succ_w^{(t+1)} m_i$.
Recall that for every $m_1,m_2 \in M$, if $m_1 \succ_w' m_2$, then all the swaps of $m_1$ are executed before all the swaps of $m_2$. In addition, since $m_{nd}$ is the most preferred man among $o(w) \setminus \set{m^*}$ according to $\succ'_w$ and $m^* \succ'_w m_{nd}$, then the following cases are not possible:
\begin{enumerate}
    \item $m_i = m^*$ and $m_j \in o_x(w, \succ_w^{(t)})$. 
    \item $m_i \in o_x(w, \succ_w^{(t)})$ and $m_j = m^*$. 
    \item $m_i = m_{nd}$ and $m_j \in o_x(w, \succ_w^{(t)})$.
    \item $m_i \in o_x(w, \succ_w^{(t)})$ and $m_j = m_{nd}$. 
    \item $m_i = m^*$ and $m_j = m_{nd}$.
\end{enumerate}
We thus need to consider only the following two cases:
\begin{enumerate}
    \item $m_i \notin o_x(w, \succ_w^{(t)})$ or $m_j \notin o_x(w, \succ_w^{(t)})$. According to the GS algorithm, a swap of such $m_i$ and $m_j$ cannot change the response of $w$ (either an acceptance or rejection). Therefore,  $o_x(w, \succ_w^{(t)}) = o_x(w, \succ_w^{(t+1)})$. 
    \item $m_i, m_j \in o_x(w, \succ_w^{(t)})\setminus \set{m^*, m_{nd}}$.  
    We use case~\ref{case:not_2_preferred_proposals} of Lemma~\ref{lemma:swapping} for this case.
    Assume to contradiction that $o_x(w, \succ_w^{(t)}) \neq o_x(w, \succ_w^{(t+1)})$. There are two possible cases:
    \begin{enumerate}
        \item There exists a man $o \in o_x(w,\succ_w^{(t)})$ such that $o \notin o_x(w,\succ_w^{(t+1)})$. By case~\ref{case:not_2_preferred_proposals} of Lemma~\ref{lemma:swapping}, $m^* = \mu_x(w,\succ_w^{(t)}) = \mu_x(w, \succ_w^{(t+1)})$. Let $\succ^o$ be $\succ_w^{(t+1)}$ such that $o$ is positioned above $m^*$.
        We can construct $\succ^o$ from $\succ_w^{(t+1)}$ by swaps of $o$. Since $o \notin o_x(w,\succ_w^{(t+1)})$, then by case~\ref{case:not_proposal} of Lemma~\ref{lemma:swapping},  
        $\mu_x(w,\succ_w^{(t+1)}) = \mu_x(w, \succ^o)$. We now swap $m_j$ and $m_i$ in $\succ^o$, and thus $m_i \succ^o m_j$ as in $\succ_w^{(t)}$.
        Let $Pre^o \subset o_x(w,\succ_w^{(t)})$ be the set of proposals that $w$ receives before she receives the proposal $o$.
        Note that all the men $o \in Pre^o$ are in the same order in $\succ_w^{(t)}$ and in $\succ^o$.
        Therefore, the response of woman $w$ is the same for all the proposals $o \in Pre^o$ and thus $o \in o_x(w,\succ^o)$. Therefore, $m^* \neq \mu_x(w,\succ^o)$, which is a contradiction. 
        
        \item There exists a man $o \notin o_x(w,\succ_w^{(t)})$ such that $o \in o_x(w,\succ_w^{(t+1)})$. Using a similar argument to case (a) above (i.e., we now construct $\succ^o$ from $\succ_w^{(t)}$) we get that in this case also $o_x(w, \succ_w^{(t)}) = o_x(w, \succ_w^{(t+1)})$

    \end{enumerate}

\end{enumerate}
\end{proof}

\subsection{Proof of Theorem~\ref{thm:alg_2_correct}}

\begin{proof}
Clearly, the algorithm runs in polynomial time since there are three loops, where the three loops together iterate at most $k^2$ times, and the running time of the GS matching algorithm is in  $O(k^2)$. 
In addition, if the algorithm returns a preference order, which is a manipulative vote for the manipulator $r$, then $m^*$ will be the match of $\hat{w}$ by the GS algorithm. 
We need to show that if there exists a preference order for the manipulator $r$ that makes $m^*$ the match of $\hat{w}$, then our algorithm will find such a preference order for $r$.
Assume that a manipulative vote, $\ell_t$, exists, which makes $m^*$ the match of $\hat{w}$. That is, $\mu(\hat{w}, \ell_t) = m^*$. 
We show that Algorithm~\ref{alg:woman_single_manip} returns $\ell_r$ in line~\ref{line:algtwo_return_ellr}.
Let $\ell_{r(1)}$ be the preference order $\ell_r$ after phase $1$ of the algorithm. $\ell_{r(2)}$ and $\ell_{r(3)}$ are defined similarly. Note that $\ell_{r(3)}$ is the preference order $\ell_r$ that is returned by the algorithm in line~\ref{line:algtwo_return_ellr}. 

Algorithm~\ref{alg:woman_single_manip} iterates over all $m_{nd} \in M \setminus \set{m^*}$, and thus there exists an iteration in which $m_{nd}$ is the second preferred proposal among $o(\hat{w}, \ell_t)$ according to $\calF(\calL \cup \set{\ell_t})$.
Let $\ell_{t(1)}$ be the preference order $\ell_t$ where $m^*$ and $m_{nd}$ are placed in the same positions as in $\ell_{r(1)}$. Note that $m^*$ and $m_{nd}$ are placed in $\ell_{r(1)}$ such that $m^*$ is preferred over $m_{nd}$ according to $\calF(\calL \cup \set{\ell_{r(1)}})$, and thus $m^*$ is preferred over $m_{nd}$ according to $\calF(\calL \cup \set{\ell_{t(1)}})$. In addition, $m_{nd}$ is positioned in $\calF(\calL \cup \set{\ell_{t(1)}})$ not lower than in $\calF(\calL \cup \set{\ell_t})$. Therefore, $m_{nd}$ is the most preferred man among $o(\hat{w}, \ell_t) \setminus \set{m^*}$ according to $\calF(\calL \cup \set{\ell_{t(1)}})$. By Corollary~\ref{corollary:m_nd}, $o(\hat{w}, \ell_t) = o(\hat{w}, \ell_{t(1)})$ and $\mu(\hat{w}, \ell_{t(1)}) = m^*$. Thus, $m_{nd}$ is the second preferred proposal among $o(\hat{w}, \ell_{t(1)})$ according to $\calF(\calL \cup \set{\ell_{t(1)}})$.

Let $\ell_{t(2)}$ be the preference order $\ell_{t(1)}$ where the men $m \not\in o(\hat{w}, \ell_{t(1)})$ are placed in the highest  positions in $\ell_{t(2)}$ without changing the positions of $m^*$ and $m_{nd}$ (similar to the positioning of the men $m \not\in o(\hat{w}, \ell_{r(2)})$ after phase $2$ of the algorithm). That is, $m^*$ is preferred over $m_{nd}$ according to $\calF(\calL \cup \set{\ell_{t(2)}})$ and $m_{nd}$ is the most preferred man among $o(\hat{w}, \ell_{t(1)}) \setminus \set{m^*}$ according to $\calF(\calL \cup \set{\ell_{t(2)}})$. We can thus use (again) Corollary~\ref{corollary:m_nd} to get that $o(\hat{w}, \ell_{t(1)}) = o(\hat{w}, \ell_{t(2)})$ and $\mu(\hat{w}, \ell_{t(2)}) = m^*$.

Recall that at the end of phase $1$ of Algorithm~\ref{alg:woman_single_manip}, $m^*$ and $m_{nd}$ are placed in $\ell_{r(1)}$ such that $m^*$ is preferred over $m_{nd}$ according to $\calF(\calL \cup \set{\ell_{r(1)}})$. In addition, $m_{nd}$ is positioned in $\calF(\calL \cup \set{\ell_{r(1)}})$ not lower than in $\calF(\calL \cup \set{\ell_{t(2)}})$, since they are placed in the same position in $\ell_{r(1)}$ and $\ell_{t(2)}$ and the other men in $\ell_{r(1)}$ get a score of $0$ from $\ell_{r(1)}$. Specifically, the men $m \in o(\hat{w}, \ell_{t(2)})$ also get a score of $0$ from $\ell_{r(1)}$ and thus $m_{nd}$ is the most preferred man among $o(\hat{w}, \ell_{t(2)}) \setminus \set{m^*}$ according to $\calF(\calL \cup \set{\ell_{r(1)}})$. 
By Corollary~\ref{corollary:m_nd}, $o(\hat{w}, \ell_{t(2)}) = o(\hat{w}, \ell_{r(1)})$ and $\mu(\hat{w}, \ell_{r(1)}) = m^*$. 
Since in phase $2$ of Algorithm~\ref{alg:woman_single_manip} we place only men $m \notin o(\hat{w}, \ell_{r(1)})$, then, we can (again) use Corollary~\ref{corollary:m_nd} to show that $o(\hat{w}, \ell_{t(2)}) = o(\hat{w}, \ell_{r(2)})$, $\mu(\hat{w}, \ell_{r(2)}) = m^*$, and $m_{nd}$ is the second preferred proposal among $o(\hat{w},\ell_{r(2)})$ according to $\calF(\calL \cup \set{\ell_{r(2)}})$.

We now show that Algorithm~\ref{alg:woman_single_manip} (lines~\ref{line:woman_begin_B}-\ref{line:woman_end_B}) can assign scores to all the men $m \in B^{nd}$ such that $\ell_r$ is a successful manipulation. For any $m \in B^{nd}$, if $s(m, \ell_{r(3)}) \leq s(m, \ell_{t(2)})$ then $s(m, \calL, \ell_{r(3)}) \leq s(m,  \calL, \ell_{t(2)})$. 
Since $(m_{nd}, \calL, \ell_{t(2)}) > (m, \calL, \ell_{t(2)})$ and $s(m_{nd},\calL,\ell_{t(2)}) = s(m_{nd},\calL,\ell_{r(3)})$ then \sloppy{$(m_{nd}, \calL, \ell_{r(3)}) > (m, \calL, \ell_{r(3)})$}. Otherwise, let $m \in B^{nd}$ be a man such that $s(m, \ell_{r(3)}) > s(m, \ell_{t(2)})$ and let $s = s(m, \ell_{r(3)})$.
By the algorithm construction, there are $s$ men from $B^{nd}$ below $m$ in $\ell_{r(3)}$.
According to the pigeonhole principle, at least one of the men from $B^{nd}$, denoted $m'$, gets a score of at least $s$ from $\ell_{t(2)}$.
That is, $s(m', \ell_{t(2)}) \geq s(m, \ell_{r(3)})$.
By the algorithm construction, all the men $m'' \in B^{nd}$ that are positioned lower than $m$ in $\ell_{r(3)}$ are positioned higher than $m$ in $\calF(\calL)$. That is, $(m', \calL) > (m, \calL)$.
However, $(m', \calL, \ell_{t(2)}) < (m_{nd}, \calL, \ell_{t(2)})$ and thus $(m, \calL, \ell_{r(3)}) < (m_{nd}, \calL, \ell_{r(3)})$. Overall, after placing the men from $B^{nd}$ in $\ell_{r(3)}$, $\forall m \in B^{nd}$, $(m, \calL, \ell_{r(3)}) < (m_{nd}, \calL, \ell_{r(3)})$. That is, $m_{nd}$ is the most preferred man among $o(\hat{w}, \ell_{r(2)}) \setminus \set{m^*}$ according to $\calF(\calL \cup \set{\ell_{r(2)}})$. In addition $(m_{nd}, \calL, \ell_{r(3)}) < (m^*, \calL, \ell_{r(3)})$ and thus by Corollary~\ref{corollary:m_nd}, $\mu(\hat{w}, \ell_{r(3)}) = m^*$.
\end{proof}

\subsection{Proof of Theorem~\ref{thm:coalitional_mnm-w}}

\begin{proof}
Clearly, the coalitional MnM-w problem is in $NP$. 
Given an instance of the Permutation Sum problem we build an instance of the coalitional MnM-w problem as follows. There are $q + 3$ men and $q + 3$ women. 
The men are: $m_1, \cdots, m_{q+2}$ and $m^*$. The women are: $w_1, \cdots, w_{q+2}$ and $\hat{w}$. 
We assume that $m^*$ is preferred over the other women according to the lexicographical tie-breaking rule. 
The $q$ men, $m_1,\ldots,m_q$, correspond to the integers $X_1,\ldots,X_q$. By Lemma $1$ from~\cite{davies2011complexity}, we can construct an election in which the non-manipulators cast votes such that:
\[
\hat{w} = (m_{q+1} ,m_1, \ldots, m_q, m^* ,m_{q + 2}),
\]
 and the corresponding scores are: 
\[
(2q + 4 + C, 2q + 4 + C - X_1 \ldots, 2q + 4 + C - X_q, C, z),
\] 
where $C$ is a constant and $z \leq C$. The preference orders of the women beside $\hat{w}$ are the same, and are as follows:
\[
(m^*,m_1,\ldots,m_{q + 2})
\] 

The preference orders of the men are the same, and are as follows:
\[
(\hat{w},w_1,\ldots,w_{q + 2})
\] 
We show that two manipulators can make the woman $\hat{w}$ the match of man $m^*$ iff the Permutation Sum problem has a solution.

($\Leftarrow$) Suppose we have two permutations $\sigma$ and $\pi$ of $1$ to $q$ such that $\sigma (i) + \pi (i) = X_i$. Let $\sigma^{-1}$ be the inverse function of $\sigma$, i.e., $i=\sigma^{-1}(m)$ if $m=\sigma(i)$. We define $\pi^{-1}(m)$  similarly.
We construct the following two manipulative votes:
\[
(m^*,m_{q+2},m_{\sigma^{-1}(q)},\ldots,m_{\sigma^{-1}(1)},m_{q+1})
\]
\[
(m^*,m_{q+2}, m_{\pi^{-1}(q)}, \ldots, m_{\pi^{-1}(1)},m_{q+1})
\]
Since $\sigma (i) + \pi (i) = X_i$
and $z \leq C$, 
the preference order of $\hat{w}$, $\calF(\calL_R, \calL)$ is: 
\[
(m^*, m_1, \ldots,m_{q+2})
\]
since the corresponding scores are:
\[
(2q+4+C,2q+4+C-X_1+X_1, \ldots, 2q+4+C-X_q+X_q,\]
\[2q+4+C, 2(q + 1) + z).
\]
Therefore, $m^*$ is the match of $\hat{w}$.

($\Rightarrow$) Assume we have a successful manipulation. 
Such a manipulation must ensure that none of the men $m_1,\ldots ,m_{q+2}$ are placed in the first positions in the preference order of $\hat{w}$.
That is, to ensure that man $m^*$ is ranked higher than man $m_{q+1}$, both manipulators have to place $m^*$ in the highest position in their preferences, and $m_{q+1}$ in the lowest position in their preferences. Thus, the score of man $m^*$ in the preference of $\hat{w}$ will be $2q+4+C$.
Let $\sigma(i)$ be a function that determines the score where the first manipulator is assigned to man $m_i$. $\pi(i)$ is defined similarly for the second manipulator. 
Since the manipulation is successful, for every $i$, $1 \leq i \leq q$, 
$$
2q+4+C-X_i +\sigma (i)+ \pi (i) \leq 2q+4+C,
$$
and thus, 
$$
\sigma (i)+ \pi (i) \leq X_i.
$$
Since $\sum_{i=1}^q X_i = q(q+1)$,
$$
\sum_{i=1}^q \sigma(i)+\pi(i) \leq q(q+1).
$$
On the other hand, since $m_{q+1}$ is placed in the lowest position by both manipulators, 
$$
\sum_{i=1}^q \sigma(i) \geq \frac{q(q+1)}{2}
$$
and
$$
\sum_{i=1}^q \pi(i) \geq \frac{q(q+1)}{2}.
$$
Therefore, 
$\sum_{i=1}^q \sigma(i)+\pi(i) = q(q+1)$, and $\sum_{i=1}^q \sigma(i) = \sum_{i=1}^q \pi(i) = \frac{q(q+1)}{2}$. That is, $\sigma$ and $\pi$ are permutations of $1$ to $q$. 
Moreover, since there is no slack in the inequalities,
$$
\sigma (i)+ \pi (i) = X_i.
$$
That is, there is a solution to the Permutation Sum problem.
\end{proof}

\subsection{Proof of Theorem~\ref{thm:coalition_manip_woman}}
In order to prove Theorem~\ref{thm:coalition_manip_woman} we use the following definitions, which correspond to the definitions we used for proving Theorem~\ref{thm:coalition_manip_man}. Given $m_{nd} \in M \setminus \set{m^*}$, let $D^{nd}_0 = \set{d_0}$, where $(d_0, \calL) > (m, \calL)$, and $d_0,m \in B^{nd}$.
For each $s=1,2,...$, let $D^{nd}_s \subseteq B^{nd}$ be $D^{nd}_s=D^{nd}_{s-1} \cup \set{d: d \in B^{nd}$ was ranked above some $d' \in D^{nd}_{s-1}$ according to $\calF(\calL, \calL_R)$ in some stage $l$, $1 \leq l \leq |R|}$. Now, let $D^{nd}= \bigcup _{0 \leq s}D^{nd}_s$. 
Note that $\forall s$ $D^{nd}_s \neq D^{nd}_{s-1}$, and $s$ does not necessarily equal $|R|$.
Let $s_\ell(m)$ be the score of man $m$ in ${\calF}(\calL, \calL_R)$ after stage $\ell$.
Let $\calL_{R-2}$ be a preference profile of $|R|-2$ manipulators.

The proof of Theorem~\ref{thm:coalition_manip_woman} relies on several Lemmata, which reminiscent Lemmata \ref{lemma:D_low}-\ref{lemma:average_fromMax}, and its general intuition is also somehow similar to the intuition behind the proof of Theorem \ref{thm:coalition_manip_man}.

We begin with a basic lemma that clarifies where Algorithm ~\ref{alg:coalition_manip_woman} places the men of $D^{nd}$.

\begin{lemma} \label{lemma:D_man_low}
Given $m_{nd} \in M \setminus \set{m^*}$, the men in $D^{nd}$ are placed in each stage $l$, $1 \leq l \leq |R|$ among the $|D^{nd}| + 1$ lowest positions.
\end{lemma}
\begin{proof}
Zuckerman et al.~\cite{zuckerman2009algorithms} 
define a set of candidates $G_W$, and show that their algorithm places them at the $|G_W|$ lowest positions (Lemma $3.5$ in \cite{zuckerman2009algorithms}). Even though our definition of the set $D^{nd}$ is slightly different, our algorithm handles the set $D^{nd}$ in almost the same way that their algorithm handles the set $G_W$.
Indeed, the algorithm might place $m_{nd}$ in one of the $|D^{nd}|$ lowest positions in phase $1$.
Thus, the men in $D^{nd}$ are placed in each stage in the $|D^{nd}|+1$ lowest positions.
\end{proof}

We now show that if there exists a successful manipulation, then 
Algorithm~\ref{alg:coalition_manip_woman} will successfully finish phase $1$. That is, the algorithm will find a man $m_{nd}$ such that $\mu(\hat{w},\calL_R)=m^*$ and $m_{nd} \in o(\hat{w},\calL_R)$ at the end of phase $1$.

\begin{lemma} \label{lemma:successful_phase_1}
Assume that a preference profile $\calL_T$ with $|R|$ manipulators exists such that $\mu(\hat{w}, \calL_T) = m^*$. Let $m'_{nd} \in o(\hat{w}, \calL_T)$ be the second most preferred man according to $\calF(\calL \cup \calL_T)$.
If the input to Algorithm~\ref{alg:coalition_manip_woman} consists of $|R|$ manipulators, then in the iteration in which $m_{nd} = m'_{nd}$, $(m_{nd}, \calL, \calL_R) < (m^*, \calL, \calL_R)$ and $s(m_{nd}, \calL, \calL_T) \leq s(m_{nd}, \calL, \calL_R)$. 
Similarly, if $\calL_T$ consists of $|R| - 2$ manipulators, and the input to Algorithm~\ref{alg:coalition_manip_woman} consists of $|R|-2$ manipulators, then in the iteration in which $m_{nd} = m'_{nd}$, $(m_{nd}, \calL, \calL_{R-2}) < (m^*, \calL, \calL_{R-2})$ and $s(m_{nd}, \calL, \calL_T) \leq s(m_{nd}, \calL, \calL_{R-2})$.
\end{lemma}

\begin{proof}
Clearly, since $m^*$ is preferred over $m'_{nd}$, then for the iteration $m_{nd} = m'_{nd}$ it satisfies $|R| (k-1) \geq gap$. Otherwise, even if each manipulator would place $m^*$ in the highest position and $m_{nd}$ in the lowest position $m_{nd}$ would be preferred over $m^*$ which contradicts the existence of $\calL_T$. 

Next we show that $m^*$ is preferred over $m_{nd}$ according to $\calF(\calL \cup \calL_R)$. We consider the two cases of the algorithm:
\begin{itemize}
    \item $|R| \geq gap$. We show that $s(m_{nd}, \calL, \calL_T) \leq s(m_{nd}, \calL, \calL_R)$ and $(m_{nd}, \calL, \calL_R) < (m^*, \calL, \calL_R)$.
    
    According to the algorithm, $s(m^*, \calL, \calL_R) = s(m^*, \calL) + \max(gap + \lceil (|R| - gap) / 2  \rceil,0)(k-1) + (|R| - \max(gap + \lceil (|R| - gap) / 2  \rceil,0))(k-2)$ and $s(m_{nd}, \calL, \calL_R) = s(m_{nd}, \calL) + \max(gap + \lceil (|R| - gap) / 2  \rceil,0)(k-2) + (|R| - \max(gap + \lceil (|R| - gap) / 2  \rceil,0))(k-1)$. 
    We show that $s(m^*, \calL, \calL_R) - s(m_{nd}, \calL, \calL_R) \geq 0$. 
    $s(m^*, \calL, \calL_R) - s(m_{nd}, \calL, \calL_R) = $
    
    $s(m^*, \calL) + 2 \max(gap + \lceil (|R| - gap) / 2  \rceil,0) - |R| - s(m_{nd}, \calL) \geq
    s(m^*, \calL) + \max(|R| + gap,0) - |R| - s(m_{nd}, \calL)$. We consider the two cases:
    \begin{enumerate}
        \item $|R| + gap > 0$. Then, $s(m^*, \calL, \calL_R) - s(m_{nd}, \calL, \calL_R) \geq s(m^*, \calL) - |R| + |R| + gap - s(m_{nd}, \calL) \geq 0$. In addition, $s(m^*, \calL, \calL_R) - s(m_{nd}, \calL, \calL_R) \leq 1$. 
        $s(m^*, \calL) + \max(|R| + gap,0) - |R| - s(m_{nd}, \calL)$
        If one of the manipulators that positioned $m^*$ above $m_{nd}$, would change the positions of $m_{nd}$ and $m^*$, then $m^*$ would be preferred over $m_{nd}$. 
        \item $|R| + gap \leq 0$. $s(m^*, \calL, \calL_R) - s(m_{nd}, \calL, \calL_R) = s(m^*, \calL) - |R| + 0 - s(m_{nd}, \calL) \geq 0$. 
        Then in each manipulator $m_{nd}$ is positioned in the highest position. 
    \end{enumerate}    
    Thus, it is not possible to give $m_{nd}$ a higher score. That is, $s(m_{nd}, \calL, \calL_T) \leq s(m_{nd}, \calL, \calL_R)$.
    

    \item $|R| < gap$. We show that 
    $(m_{nd}, \calL, \calL_R) < (m^*, \calL, \calL_R)$. 
    
    According to the algorithm, $s(m^*, \calL, \calL_R) = s(m^*, \calL) + |R| (k - 1)$
    and $s(m_{nd}, \calL, \calL_R) = s(m_{nd}, \calL) + (s_{m_{nd}} \mod |R|) \lceil \frac{|R|(k-1) - gap}{|R|} \rceil + (|R| - (s_{m_{nd}} \mod |R|)) \lfloor \frac{|R|(k-1) - gap}{|R|} \rfloor$. 
    
    $s(m^*, \calL, \calL_R) - s(m_{nd}, \calL, \calL_R) = 
    s(m^*, \calL) + |R|(k - 1) 
    - s(m_{nd}, \calL) - (s_{m_{nd}} \mod |R|) \lceil \frac{|R|(k-1) - gap}{|R|} \rceil - (|R| - (s_{m_{nd}} \mod |R|)) \lfloor \frac{|R|(k-1) - gap}{|R|} \rfloor \geq 0$.

    Since, $|R|(k - 1) \geq gap$.

\end{itemize}

Since there exists $m'_{nd} \in o(\hat{w}, \calL_T)$, there is no option to award $m_{nd}$ a higher score such that it will not overcome $m^*$ according to $\calF(\calL \cup \calL_R)$ then according to line~\ref{line:pass_phase_1}, $m_{nd} \in o(\hat{w}, \calL_R)$ after phase $1$ and $s(m_{nd}, \calL, \calL_R) \geq s(m_{nd}, \calL, \calL_T)$.

\end{proof}

\begin{corollary}\label{corrolary:successful_phase_1}
Assume that a preference profile $\calL_T$ with $|R|$ manipulators exists such that $\mu(\hat{w}, \calL_T) = m^*$. Let $m'_{nd} \in o(\hat{w}, \calL_T)$ be the second most preferred man according to $\calF(\calL \cup \calL_T)$.
If the input to Algorithm~\ref{alg:coalition_manip_woman} consists of $|R|$ manipulators, then in the iteration in which $m_{nd} = m'_{nd}$ the algorithm will successfully finish phase 1. Similarly, if $\calL_T$ consists of $|R| - 2$ manipulators, and the input to Algorithm~\ref{alg:coalition_manip_woman} consists of $|R|-2$ manipulators, then in the iteration in which $m_{nd} = m'_{nd}$ the algorithm will successfully finish phase 1.
\end{corollary}

\begin{proof}
According to Lemma~\ref{lemma:successful_phase_1}, $(m_{nd}, \calL, \calL_R) < (m^*, \calL, \calL_R)$ and $s(m_{nd}, \calL, \calL_T) \leq s(m_{nd}, \calL, \calL_R)$. And therefore in the end of phase 1, $\forall m \in o(\hat{w}, \calL_R)$, $(m, \calL, \calL_R) < (m_{nd}, \calL, \calL_R)$. Since $\mu(\hat{w}, \calL_T) = m^*$ and $m_{nd} \in o(\hat{w}, \calL_T)$, then by Corollary~\ref{corollary:m_nd}, $\mu(\hat{w}, \calL_R) = m^*$ and $m_{nd} \in o(\hat{w}, \calL_R)$ in the end of phase 1. That is, Algorithm~\ref{alg:coalition_manip_woman} successfully finishes phase 1. Clearly, the proof is similar for manipulation with $|R|-2$ manipulators.
\end{proof}

We now show the relation between the score of $m_{nd}$ and the average score of the men in $D^{nd}$, when there are $|R|-2$ manipulators.

\begin{lemma} \label{lemma:woman_times_algorithm_returns_false}
Assume that a preference profile $\calL_T$ with $|R| - 2$ manipulators exists such that $\mu(\hat{w}, \calL_T) = m^*$. Let $m'_{nd} \in o(\hat{w}, \calL_T)$ be the second most preferred man according to $\calF(\calL \cup \calL_T)$. 
Assume that the input to Algorithm~\ref{alg:coalition_manip_woman} consists of $|R|-2$ manipulators, and let $q(D^{nd})$ be the average score of men in $D^{nd}$. That is, $q(D^{nd}) = {\frac{1}{|D^{nd}|}}\sum_{d \in D^{nd}}s(d, \calL, \calL_{R-2})$. Then, in the iteration in which $m_{nd} = m'_{nd}$, $s(m_{nd}, \calL, \calL_{|R|-2}) \geq q(D^{nd})$.
\end{lemma}

\begin{proof}

Let $q^T(D^{nd})={\frac{1}{|D^{nd}|}} \sum_{d \in D^{nd}}s(d, \calL, \calL_T)$ be the average score of men in $D^{nd}$ in $(\calL, \calL_T)$.
As a corollary of Lemma~\ref{lemma:D_man_low},
if $|R| \geq gap$, then all the men in $D^{nd}$ are placed among the $|D^{nd}|$ lowest positions in the preference profile of each manipulator. That is, 

$q(D^{nd}) = {\frac{1}{|D^{nd}|}}(\sum_{d \in D^{nd}}s(d, \calL)+ \sum_{j = 1}^{|R|-2} \sum_{i = 0}^{|D^{nd}|-1}i) \leq
{\frac{1}{|D^{nd}|}} \sum_{d \in D^{nd}}s(d, \calL, \calL_{T})=q^T(D^{nd})
$.

Now, if $|R| - 2 < gap$, but $\lfloor \frac{(|R|-2)(k-1) - gap}{|R|-2} \rfloor \geq |D^{nd}|$ then again all the men in $D^{nd}$ are placed among the $|D^{nd}|$ lowest positions in the preference profile of each manipulator, and $q(D^{nd}) \leq q^T(D^{nd})$.
Finally, consider the case in which $|R| - 2 < gap$, and $\lfloor \frac{(|R|-2)(k-1) - gap}{|R| - 2} \rfloor < |D^{nd}|$. Note that for each $\ell_r \in \calL_r$, only men from $D^{nd}$ are placed below $m_{nd}$, and all the men from $D^{nd}$ are placed among the $|D^{nd}|+1$ lowest positions.
Therefore, if for every $1 \leq i \leq |R| - 2$, $m_{nd}$ is positioned in $\ell_i \in \calL_T$ at the same position as in $\ell_i \in \calL_R$ or at a lower position, then obviously $q(D^{nd}) \leq q^T(D^{nd})$.
If it is not the case, then there is a subset of manipulators $\calL^+_T \subset \calL_T$, and a corresponding subset $\calL^+_R \subset \calL_{R-2}$, such that $m_{nd}$ is positioned in $\ell_i \in \calL^+_T$ at a higher position than in $\ell_i \in \calL^+_R$. 
That is, $\Delta = s(m_{nd},\calL^+_T)-s(m_{nd},\calL^+_R)$ is greater than $0$, and thus $\sum_{d\in D^{nd}} s(d,\calL^+_R) - s(d,\calL^+_T)$ is at most $\Delta$.
However, since $s(m_{nd}, \calL, \calL_T) \leq s(m_{nd}, \calL, \calL_{R-2})$ according to Lemma~\ref{lemma:successful_phase_1}, then there must be a non-empty subset of manipulators $\calL^-_T \subset \calL_T$, and a corresponding subset $\calL^-_R \subset \calL_{R-2}$, such $m_{nd}$ is positioned in $\ell_i \in \calL^-_T$ at a lower position than in $\ell_i \in \calL^-_R$. Moreover, $s(m_{nd},\calL^-_R)-s(m_{nd},\calL^-_T) \geq \Delta$ and thus 
$\sum_{d\in D^{nd}} s(d,\calL^-_T) - s(d,\calL^-_R)$ is at least $\Delta$. Therefore, $q(D^{nd}) \leq q^T(D^{nd})$. 

Note that $m'_{nd} \in o(\hat{w}, \calL_T)$ is the second most preferred man according to $\calF(\calL \cup \calL_T)$, and $\mu(\hat{w}, \calL_T) = m^*$. Therefore, according to Lemma~\ref{lemma:successful_phase_1}, in the iteration in which $m_{nd} = m'_{nd}$, $(m_{nd}, \calL, \calL_{R-2}) < (m^*, \calL, \calL_{R-2})$ and $s(m_{nd}, \calL, \calL_T) \leq s(m_{nd}, \calL, \calL_{R-2})$. 
In addition, after phase $1$ of the algorithm, $M \setminus \set{m_{nd}, m^*}$ are not placed yet in $\calL_{R-2}$, and thus at the beginning of phase $2$, $B^{nd} = o(\hat{w}, \calL_T) \setminus \set{m^*, m_{nd}}$ according to Corollary~\ref{corollary:m_nd}. 
Since $\mu(\hat{w}, \calL_T) = m^*$, 
then for all $d \in o(\hat{w}, \calL_T) \setminus \set{m^*, m_{nd}}$,
$s(d, \calL, \calL_T) \leq s(m_{nd}, \calL, \calL_T)$. Since $D^{nd} \subseteq B^{nd}$, then $q^T(D^{nd}) \leq s(m_{nd}, \calL, \calL_T)$.
Since we showed that $q(D^{nd}) \leq q^T(D^{nd})$ and $s(m_{nd}, \calL, \calL_T) \leq s(m_{nd}, \calL, \calL_{R-2})$, we get that $q(D^{nd}) \leq s(m_{nd}, \calL, \calL_{R-2})$.
\end{proof}



The following lemma shows the relation between the score of $m_{nd}$ and the maximum score of a man in $D^{nd}$, when there are $|R|$ manipulators. In essence, the lemma characterizes the settings in which Algorithm~\ref{alg:coalition_manip_woman} finds a successful manipulation.

\begin{lemma}\label{lemma:woman_times_algorithm_succeeds}
Assume that a preference profile $\calL_T$ with $|R|$ manipulators exists such that $\mu(\hat{w}, \calL_T) = m^*$. Let $m'_{nd} \in o(\hat{w}, \calL_T)$ be the second most preferred man according to $\calF(\calL \cup \calL_T)$.
If the input to Algorithm~\ref{alg:coalition_manip_woman} consists of $|R|$ manipulators, and in the iteration in which $m_{nd} = m'_{nd}$, $\max_{d \in D^{nd}}\set{s(d, \calL, \calL_R)} < s(m_{nd}, \calL, \calL_R)$, then Algorithm~\ref{alg:coalition_manip_woman} finds a successful manipulation.
\end{lemma}

\begin{proof}
Since there exists a preference profile $\calL_T$ that ensures that $\mu(\hat{w}, \calL_T) = m^*$.
Then, according to Lemma~\ref{lemma:successful_phase_1} the algorithm can place $m^*$ and $m_{nd}$ such that in the end of phase 1 $\mu(\hat{w}, \calL_R) = m^*$ and $s(m_{nd}, \calL, \calL_T) \leq s(m_{nd}, \calL, \calL_R)$.

Even though the set $D^{nd}$ in the proof of Theorem~\ref{thm:coalition_manip_woman} is defined slightly differently than in the proof of Theorem~\ref{thm:coalition_manip_man}, and this Lemma is slightly different than Lemma~\ref{lemma:times_algorithm_succeeds}, the proof of this lemma is essentially identical to the proof of Lemma~\ref{lemma:times_algorithm_succeeds}. And therefore $\forall m \in o(\hat{w}, \calL_R) \setminus \set{m^*, m_{nd}}$ in the end of phase 1, $(m, \calL, \calL_R) < (m_{nd}, \calL, \calL_R)$ in the end of phase 2, and according to Corollary~\ref{corollary:m_nd} $\mu(\hat{w}, \calL_R) = m^*$ in the end of phase 2.
Thus, Algorithm~\ref{alg:coalition_manip_woman} will find a successful manipulation.
\end{proof}

Next, we show that the highest score of a man from $D^{nd}$, when there are $|R|$ manipulators, is not much higher than the average score of the men in $D^{nd}$, when there are $|R|-2$ manipulators. We thus first show that the scores of men in $D^{nd}$ according to $\calF(\calL \cup \calL_{R-2})$ and  $\calF(\calL \cup \calL_{R-1})$ are almost $1$-dense, as captured by the following definition. 

\begin{definition}
\label{def:almost_1_dense}
A finite non-empty set of integers $B$ is called almost $1$-dense if, when sorting the set in a non-increasing order $b_1 \geq b_2 \geq \dots \geq b_i$ (such that $ \set{b_1, \dots, b_i}=B$), there exists at most one index, $j$, such that $1 \leq j \leq i-1$, $b_{j+1} = b_j - 2$ and $\forall t \neq j, 1 \leq t \leq i-1$, $b_{t+1} \geq b_t - 1$ holds.
\end{definition}

\begin{lemma} \label{lemma:subset_almost_1_dense}
Let $D^{nd} = \bigcup_{0 \leq s} D^{nd}_s$, as before. Then for all $s \geq 1$ and $d \in D^{nd}_s \setminus D^{nd}_{s-1}$ there exist $d' \in D^{nd}_{s-1}, X \subseteq B^{nd}$ and $j, 0 \leq j \leq |R|$, s.t. $\set{s_j(d), s_j(d')} \cup s_j(X)$ is almost $1$-dense.
\end{lemma}

\begin{proof}
Let $s \geq 1$ and $d \in D^{nd}_s \setminus D^{nd}_{s-1}$. By definition, there exist $d' \in D^{nd}_{s-1}$ and a minimal $j, 0 \leq j \leq |R|$, such that $d$ was ranked below $d'$ in stage $j$, that is $s(d, \ell_j) < s(d', \ell_j)$. We distinguish between two cases:
    case $1$: $j > 1$. In this case $d$ was ranked above $d'$ in stage $j-1$. So we have:
\begin{equation} \label{eq:stage_j-1}
    s_{j-1}(d) \geq s_{j-1}(d')
\end{equation}

\begin{equation} \label{eq:stage_j-2}
    s_{j-2}(d) \leq s_{j-2}(d')
\end{equation}
$d$ was ranked above $d'$ in stage $j-1$, and hence $s(d, \ell_{j-1}) > s(d', \ell_{j-1})$. 
Denote $g = s(d, \ell_{j-1}) - s(d', \ell_{j-1})$ - 1. We distinguish into two cases:

$\bullet$ Let $d' = d_0, d_1, ..., d_{g+1} = d$ be the men that got in stage $j-1$ the points $s(d', \ell_{j-1}), s(d', \ell_{j-1}) + 1, ..., s(d', \ell_{j-1}) + g + 1$, respectively. In this case $m_{nd}$ was not placed between the men $\set{d_0, d_1, ..., d_{g+1}}$, but higher than $d$ or lower than $d'$ in stage $j-1$.
This case corresponds to the $1$-dense case as proved in \citet{zuckerman2009algorithms}. Since $1$-dense is almost $1$-dense by definition. 

$\bullet$ Let $d' = d_0, d_1, ..., d_g = d$ be the men that got in stage $j-1$ the points $s(d', \ell_{j-1}), s(d', \ell_{j-1}) + 1, ..., s(d', \ell_{j-1}) + g + 1$, respectively. In this case $m_{nd}$ was placed between the men $\set{d_0, d_1, ..., d_g}$, but higher than $d'$ and lower than $d$ in stage $j-1$.
Our purpose is to show that $\set{s_{j-1}(d_0), ..., s_{j-1}(d_g)}$ is almost $1$- dense. By definition of the algorithm, 
\begin{equation} \label{eq:seq_stage_j-2}
    s_{j-2}(d_0) \geq s_{j-2}(d_1) \geq ... \geq s_{j-2}(d_g)    
\end{equation}
Let $m_{nd}$ be placed between $d_i$ and $d_{i+1}$, $i \in \set{0, ..., g}$. That is, $s(d_i, \ell_{j-1}) + 2 = s(d_{i+1}, \ell_{j-1})$. Denote $u_t = s_{j-2}(d_t) + s(d', \ell_{j-1})$ for $0 \leq t \leq g$. Then
\begin{equation} \label{eq:dt_score_stage_j-1_prev_i}
    \forall t, 0 \leq t \leq i, s_{j-1}(d_t) = u_t + t.
\end{equation}

\begin{equation} \label{eq:dt_score_stage_j-1_after_i}
    \forall t, i + 1 \leq t \leq g, s_{j-1}(d_t) = u_t + t + 1.
\end{equation}

We need to show that $\set{s_{j-1}(d_t) | 0 \leq t \leq g}$ is almost $1$-dense. It is enough to show that:

\begin{enumerate}
    \item For all $t$, $0 \leq t < i$, if $u_t + t < u_0$, then there exists $t'$, $t < t' \leq i$ s.t. $u_t + t < u_{t'} + t' \leq u_t + t + 1$
    \item For all $t$, $i+1 \leq t < g$, if $u_t + t + 1 < u_0$, then there exists $t'$, $t < t' \leq g$ s.t. $u_t + t < u_{t'} + t' \leq u_t + t + 1$
    \item For all $t$, $0 < t \leq i$, if $u_t + t > u_0$, then there exists $t'$, $0 \leq t' < t$ s.t. $u_t + t - 1 \leq u_{t'} + t' < u_t + t$
    \item For all $t$, $i+1 < t \leq g$, if $u_t + t + 1 > u_0$, then there exists $t'$, $i+1 \leq t' < t$ s.t. $u_t + t - 1 \leq u_{t'} + t' < u_t + t$
    \item  For $t = i$ and for $y = i + 1$ that either $u_t + t < u_0$ or $u_y + y + 1 > u_0$. There are $1$ of $2$ possible cases:
    \begin{enumerate}
        \item If $u_t + t < u_0$, then there exists $t'$, $t < t' \leq g$ s.t. $u_t + t < u_{t'} + t' + 1 = u_t + t + 2$ and, if $u_y + y + 1 > u_0$, then there exists $y'$, $0 \leq y' < y$ s.t. $u_y + y + 1 - 1 \leq u_{y'} + y' < u_y + y + 1$
        \item If $u_t + t < u_0$, then there exists $t'$, $t < t' \leq g$ s.t. $u_t + t < u_{t'} + t' + 1 = u_t + t + 1$ and, if $u_y + y + 1 > u_0$, then there exists $y'$, $0 \leq y' < y$ s.t. $u_y + y + 1 - 2 \leq u_{y'} + y' < u_y + y + 1$
    \end{enumerate}

\end{enumerate}

Proof of $1$: From (\ref{eq:seq_stage_j-2}) we get
\begin{equation} \label{eq:seq_ut}
    u_0 \geq ... \geq u_g
\end{equation}

Also from (\ref{eq:stage_j-1}) and (\ref{eq:dt_score_stage_j-1_prev_i}) we have $u_0 \leq u_g + g + 1$. 
Let $0 \leq t < i$ s.t. $u_t + t < u_0$. Let us consider the sequence $u_t + t, u_{t+1} + t +1, ..., u_i + i$. 
Since $u_t + t < u_0 \leq u_i + i$, it follows that there is a minimal index $t'$, $t < t' \leq i$ s.t. $u_t + t < u_{t'} + t'$.
Then $u_{t'-1}+t'-1 \leq u_t + t$, and thus
\begin{equation} \label{eq:t<i_ut'-1_and_ut}
    u_{t'-1}+t' \leq u_t + t + 1.
\end{equation}
From (\ref{eq:seq_ut}) $u_{t'} \leq u_{t'-1}$, and then
\begin{equation} \label{eq:t<i_ut'-1_and_ut'}
    u_{t'} + t' \leq u_{t'-1} + t'.
\end{equation}
Combining (\ref{eq:t<i_ut'-1_and_ut}) and (\ref{eq:t<i_ut'-1_and_ut'}) together, we get $u_{t'} + t' \leq u_t + t + 1$. This concludes the proof of $1.$ for $0 \leq t < i$.

Let $i+1 \leq t < g$ s.t.  $u_t + t + 1 < u_0$. Let us consider the sequence $u_t + t + 1, u_{t+1} + t +2, ..., u_g + g + 1$.
Since $u_t + t + 1 < u_0 \leq u_g + g + 1$, it follows that there is a minimal index $t'$, $t < t' \leq g$ s.t. $u_t + t + 1 < u_{t'} + t' + 1$.
Then 
\begin{equation} \label{eq:t>i_ut'-1_and_ut}
    u_{t'-1}+t' \leq u_t + t + 1.
\end{equation}
From (\ref{eq:seq_ut}) $u_{t'} \leq u_{t'-1}$, and then
\begin{equation} \label{eq:t>i_ut'-1_and_ut'}
    u_{t'} + t' \leq u_{t'-1} + t'.
\end{equation}
Combining (\ref{eq:t>i_ut'-1_and_ut}) and (\ref{eq:t>i_ut'-1_and_ut'}) together, we get $u_{t'} + t' \leq u_t + t + 1$. This concludes the proof of $1.$

The proof of $2.$ is similar to the proof of $1.$, by replacing the use of Equation~(\ref{eq:dt_score_stage_j-1_after_i}) instead of Equation~(\ref{eq:dt_score_stage_j-1_prev_i}). Rest of the equations are the same.

The proof of $3.$ is analogous to the proof of $1.$, by choosing $t'$ to be the maximal index such that $u_{t'} + t' < u_t + t$.

The proof of $4.$ is analogous to the proof of $2.$, by choosing $t'$ to be the maximal index such that $u_{t'} + t' < u_t + t$.

The proof of $5.$ considers the position of $m_{nd}$ in $\ell_j$ between $d_i$ and $d_{i+1}$.
The proof of $5.(a)$, 
for the first part let us consider the sequence $u_t + t, u_{t+1} + t + 2, ..., u_g + g + 1$. 
Since $u_t + t < u_0 \leq u_g + g + 1$, it follows that there is a minimal index $t'$, $t < t' \leq g$ s.t. $u_t + t < u_{t'} + t' + 1$.
Then $u_{t'-1}+t' \leq u_t + t$, and thus we conclude with the same equation as Equation~\ref{eq:t<i_ut'-1_and_ut}.

From (\ref{eq:seq_ut}) $u_{t'} \leq u_{t'-1}$, and then we have again Equation~\ref{eq:t<i_ut'-1_and_ut'}

Combining (\ref{eq:t<i_ut'-1_and_ut}) and (\ref{eq:t<i_ut'-1_and_ut'}) together, we get $u_{t'} + t' < u_t + t + 1$. 
The second part of the proof, is of the proof of $3$. This concludes the proof of $5.(a)$

The proof of $5.(b)$ the first part is of the proof of $1$.
The second part is analogous to $5.(a)$, by choosing $t'$ to be the maximal index such that $u_{t'} + t' < u_t + t$.

Case $2$: $j=1$. We proceed by essentially reducing this case to Case $1$. In Case $2$ we have that $s \geq 2$, because otherwise, if $s=1$, then $d'\in D^{nd}_0$: therefore $s_0(d) \geq s_0(d')=\max_{b\in B^{nd}}\set{s_0(b)}\rightarrow d \in D^{nd}_0$, a contradiction. $d' \notin D^{nd}_{s-2}$, because otherwise, by definition, $d \in D^{nd}_{s-1}$. Therefore there exists $d'' \in D^{nd}_{s-2}$ s.t. $d'$ was ranked below $d''$ in some stage $j'$, i.e., $s_{j'-1}(d) \leq s_{j'-1}(d'')$. By combining the last arguments, we get that $s_{j'-1}(d) \leq s_{j'-1}(d')$.
Let $j_0$ be minimal s.t. $s_{j_0}(d) \geq s_{j_0}(d')$. As in stage $1$, $d$ was ranked below $d'$, it holds that $s_0(d) \geq s_0(d')$. If $j_0 = 0$ then $s_0(d) = s_0(d')$, hence $\set{s_0(d), s_0(d')}$ is $0$-dense, and in particular $1$-dense and almost $1$-dense.

Otherwise $(j_0 \neq 0)$ it holds that $s_{j_0-1}(d) > s_{j_0-1}(d')$ by the minimality of $j_0$. So, we have that
\begin{equation}
    s_{j_0}(d') \geq s_{j_0}(d)
\end{equation}
and 
\begin{equation}
    s_{j_0-1}(d') \leq s_{j_0-1}(d)
\end{equation}
These two inequalities are analogous to (\ref{eq:stage_j-1}) and (\ref{eq:stage_j-2}), with $j-1$ replaced by $j_0$, and the roles of $d$ and $d'$ exchanged. From this point we can proceed exactly as in Case $1$, keeping in mind these cosmetic changes.
\end{proof}

\begin{lemma}\label{lemma:set_almost_1_dense}
    Let $H \subseteq B$ such that $s_j(H)$ is almost $1$-dense for some $0 \leq j \leq |R|$. Then there exists $H'$, $H \subseteq H' \subseteq B$ such that $s_{|R|}(H')$ is almost $1$-dense.
\end{lemma}

\begin{proof}
    The proof is the same as in~\cite{zuckerman2009algorithms}, proof of Lemma 3.11.
\end{proof}

\begin{lemma}\label{lemma:Dnd_almost_1_dense}
    The set $D^{nd}$ is almost $1$-dense.
\end{lemma}
\begin{proof}
    The proof is the similar to the one in~\cite{zuckerman2009algorithms}, proof of Lemma 3.12.
    We Note that here for two almost $1$-dense sets $A, A'$, if $A \cap A' \neq \emptyset$, then $A \cup A'$ is also almost $1$-dense, since the almost $1$-dense occurs after a manipulator places $m_{nd}$ between to men from $B$. Since this can happen only once in each manipulator. Clearly, either the set $A$ or the set $A'$ is $1$-dense. Otherwise $A \cap A'$ is already almost $1$-dense, and this concludes, that $A \cup A'$ is almost $1$-dense.
\end{proof}

\begin{lemma} \label{lemma:woman_average_fromMax}
Assume that the input to Algorithm~\ref{alg:coalition_manip_woman} consists of  $|R|$ manipulators. If there exists $m_{nd} \in M \setminus \set{m^*}$ such that the condition in Line~\ref{line:pass_phase_1} does not hold, then $\max_{d \in D^{nd}}\set{s(d, \calL, \calL_R)} \leq q(D^{nd}) + \frac{3|D^{nd}|}{2} + 1$.
\end{lemma}

\begin{proof}
Recall that $q(D^{nd})$ is the average score of men in $D^{nd}$ when there are $|R|-2$ manipulators.
Let $q(D^{nd},R)$ be the average score of men in $D^{nd}$ when there are $|R|$ manipulators. That is, $q(D^{nd},R) = {\frac{1}{|D^{nd}|}}\sum_{d \in D^{nd}}s(d, \calL, \calL_R)$. We define $q(D^{nd},R-1)$ similarly.
According to Lemma~\ref{lemma:Dnd_almost_1_dense}, the set $\set{s_{|R|}(d):d \in D^{nd}}$ is almost $1$-dense, and therefore $\max_{d \in D^{nd}}\set{s(d, \calL, \calL_R)} \leq q(D^{nd},R) + \frac{|D^{nd}|}{2}$. According to Lemma~\ref{lemma:D_low}, in each stage each man from $D^{nd}$ is placed among the $|D^{nd}|+1$ lowest positions. Therefore, $q(D^{nd},R) + \frac{|D^{nd}|}{2} \leq q(D^{nd},R-1) + |D^{nd}| + \frac{1}{2}$. Using this Lemma again we get $q(D^{nd},R-1) + |D^{nd}| + \frac{1}{2} \leq q(D^{nd},R-2) +  \frac{3|D^{nd}|}{2} + 1$
\end{proof}

Finally, we show that by adding two manipulators, we can increase the score of $m_{nd}$ by at least $2k-3$.
\begin{lemma} \label{lemma:additional_manipulator}
Assume that the input to Algorithm~\ref{alg:coalition_manip_woman} consists of  $|R|-2$ manipulators. If there exists $m_{nd} \in M \setminus \set{m^*}$ such that the condition in Line~\ref{line:gap_too_big} does not hold, then when adding $2$ additional manipulators, $s(m_{nd}, \calL, \calL_{R-2}) + 2k - 3 \leq s(m_{nd}, \calL, \calL_R)$.
\end{lemma}
\begin{proof}
Clearly, if in line~\ref{line:gap_too_big}, $(|R|-2)(k-1) \geq gap$, then $|R|(k-1) \geq gap$.
Let $t = gap + \lceil (|R| - gap) / 2 \rceil$  and let $t_{-2} = gap + \lceil (|R| - 2 - gap) / 2 \rceil$. Note that $t-t_{-2} = 1$.
There are three possible cases: 

\begin{itemize}
    \item $|R| - 2 \geq gap$ and $|R| > gap$. 

When there are $|R|$ manipulators, according to Algorithm~\ref{alg:coalition_manip_woman},  
$s(m_{nd}, \calL, \calL_R) = s(m_{nd}, \calL) + (k-2)\max(t,0) + (k-1)(|R| - \max(t,0))$. 
Similarly, when there are $|R| - 2$ manipulators, according to Algorithm~\ref{alg:coalition_manip_woman},
$s(m_{nd}, \calL, \calL_{R-2}) = s(m_{nd}, \calL) + (k-2)\max(t_{-2},0) + (k-1)(|R| - 2 - \max(t_{-2},0))$. 

Consider the following three possible cases: 
\begin{enumerate}
    \item $t_{-2} \geq 0$ and $t \geq 0$.
    
    That is, $s(m_{nd}, \calL, \calL_R) = s(m_{nd}, \calL) + (k-2)t + (k-1)(|R| - t) = s(m_{nd}, \calL) -t + (k-1)|R|$.
    Similarly, $s(m_{nd}, \calL, \calL_{R-2}) = s(m_{nd}, \calL) + (k-2)t_{-2} + (k-1)(|R| - 2 - t_{-2}) = s(m_{nd}, \calL) -t_{-2} + (k-1)|R|-2(k-1)$.
    Therefore, 
    $s(m_{nd}, \calL, \calL_R) - s(m_{nd}, \calL, \calL_{R-2}) = -1 + 2(k-1) = 2k - 3$.
    
    \item $t_{-2} < 0$ and $t \geq 0$.
    
    That is, $max(t_{-2},0) = 0$, and thus $s(m_{nd}, \calL, \calL_{R-2}) = s(m_{nd}, \calL) + (k-1)(|R| - 2)$.
    Therefore, $s(m_{nd}, \calL, \calL_R) - s(m_{nd}, \calL, \calL_{R-2}) = -t + 2(k-1)$. Recall that $t-t_{-2}=1$, and since we assume that $t_{-2} < 0$ and $t \geq 0$, it must be that $t=0$. Therefore, $s(m_{nd}, \calL, \calL_R) - s(m_{nd}, \calL, \calL_{R-2}) = 2(k-1) \geq 2k-3$.
    
    \item $t_{-2} < 0$ and $t < 0$.
    
    That is, $max(t,0) = max(t_{-2},0) = 0$, and thus $s(m_{nd}, \calL, \calL_R) = s(m_{nd}, \calL) + (k-1)|R|$.
    Therefore, $s(m_{nd}, \calL, \calL_R) - s(m_{nd}, \calL, \calL_{R-2}) = 2(k-1) \geq 2k-3$.
\end{enumerate}

\item $|R| - 2 < gap$ and $|R| \geq gap$.

That is, $gap = |R|$ or $gap = |R| - 1$. Therefore, $t=|R|$.
Recall that $s(m_{nd}, \calL, \calL_R) = s(m_{nd}, \calL) + (k-2)\max(t,0) + (k-1)(|R| - \max(t,0))$. That is, $s(m_{nd}, \calL, \calL_R) = s(m_{nd}, \calL) + (k-2)|R|$.
In addition, according to Algorithm~\ref{alg:coalition_manip_woman}, $s(m_{nd}, \calL, \calL_{R-2}) = s(m_{nd}, \calL) + (\lceil \frac{(|R| - 2)(k-1) - gap}{|R| - 2} \rceil)(((|R| - 2)(k-1) - gap) \mod (|R| - 2)) + (\lfloor \frac{(|R| - 2)(k-1) - gap}{|R|-2} \rfloor)(|R| - 2 - (((|R| - 2)(k-1) - gap) \mod (|R| - 2)))$.
Clearly, $(\lceil \frac{(|R| - 2)(k-1) - gap}{|R| - 2} \rceil)(((|R| - 2)(k-1) - gap) \mod (|R| - 2)) + (\lfloor \frac{(|R| - 2)(k-1) - gap}{|R|-2} \rfloor)(|R| - 2 - (((|R| - 2)(k-1) - gap) \mod (|R| - 2))) = (|R| - 2)(k-1) - gap$. Therefore, $s(m_{nd}, \calL, \calL_{R-2}) = s(m_{nd}, \calL) + (|R| - 2)(k-1) - gap$.

Therefore, $s(m_{nd}, \calL, \calL_R) - s(m_{nd}, \calL, \calL_{R-2}) = (k-2)|R| - (|R| - 2)(k-1) + gap = -|R| + 2k - 2 + gap$. 
Since $gap = |R|$ or $gap = |R| - 1$, then $s(m_{nd}, \calL, \calL_R) - s(m_{nd}, \calL, \calL_{R-2}) \geq 2k - 3$.

\item $|R|-2 < gap$ and $|R| < gap$.

According to Algorithm~\ref{alg:coalition_manip_woman}, $s(m_{nd}, \calL, \calL_R) = s(m_{nd}, \calL) + (\lceil \frac{(|R|)(k-1) - gap}{|R|} \rceil)((|R| (k-1) - gap) \mod |R|) + (\lfloor \frac{|R| (k-1) - gap}{|R|} \rfloor)(|R| - (|R| (k-1) - gap) \mod |R|))$. Clearly, $(\lceil \frac{(|R|) (k-1) - gap}{|R|} \rceil)((|R| (k-1) - gap) \mod |R|) + (\lfloor \frac{|R| (k-1) - gap}{|R|} \rfloor)(|R| - (|R| (k-1) - gap) \mod |R|)) = (|R|) (k-1) - gap$. Therefore, $s(m_{nd}, \calL, \calL_R) = s(m_{nd}, \calL) + (|R|)(k-1) - gap$.

In addition, according to Algorithm~\ref{alg:coalition_manip_woman}, $s(m_{nd}, \calL, \calL_{R-2}) = s(m_{nd}, \calL) + (\lceil \frac{(|R| - 2) (k-1) - gap}{|R| - 2} \rceil)(((|R| - 2) (k-1) - gap) \mod (|R| - 2)) + (\lfloor \frac{(|R| - 2) (k-1) - gap}{|R|-2} \rfloor)(|R| - 2 - (((|R| - 2) (k-1) - gap) \mod (|R| - 2)))$.
Clearly, $(\lceil \frac{(|R| - 2) (k-1) - gap}{|R| - 2} \rceil)(((|R| - 2) (k-1) - gap) \mod (|R| - 2)) + (\lfloor \frac{(|R| - 2) (k-1) - gap}{|R|-2} \rfloor)(|R| - 2 - (((|R| - 2) (k-1) - gap) \mod (|R| - 2))) = (|R| - 2) (k-1) - gap$. Therefore, $s(m_{nd}, \calL, \calL_{R-2}) = s(m_{nd}, \calL) + (|R| - 2) (k-1) - gap$.

Therefore, $s(m_{nd}, \calL, \calL_R) - s(m_{nd}, \calL, \calL_{R-2}) = s(m_{nd}, \calL) + (|R|) (k-1) - gap - s(m_{nd}, \calL) - (|R| - 2) (k-1) + gap = 2k-2 > 2k-3$

\end{itemize}

Overall, $s(m_{nd}, \calL, \calL_R) - s(m_{nd}, \calL, \calL_{R-2}) \geq 2k-3$.
\end{proof}
Now we can prove the theorem.
\begin{proof} [Proof of Theorem~\ref{thm:coalition_manip_woman}]
Clearly, if Algorithm~\ref{alg:coalition_manip_woman} returns a preference profile $\calL_R$, then it is a successful manipulation that will make $m^*$ the match of $\hat{w}$.
Suppose that a preference profile $\calL_T$ exists that makes $m^*$ the match of $\hat{w}$ with $|R|-2$ manipulators. 
Let $m'_{nd} \in o(\hat{w},\calL_T)$ be the second most preferred man according to $\calF(\calL \cup \calL_T)$.
We divide into two cases:
\begin{enumerate}
    \item $B^{nd} = \emptyset$. According to Corollary~\ref{corrolary:successful_phase_1}, in the iteration in which $m_{nd} = m'_{nd}$, the algorithm successfully finishes phase $1$. Since $B^{nd} = \emptyset$, then according to Corollary~\ref{corollary:m_nd}, in the iteration in which $m_{nd} = m'_{nd}$, the algorithm successfully finishes phase $2$. That is, the algorithm will find a preference profile that will make $w^*$ the match of $\hat{m}$ with $|R| - 2$ manipulators.
    \item $B^{nd} \neq \emptyset$. Clearly, a preference profile $\calL_T'$ exists that makes $m^*$ the match of $\hat{w}$ with $|R|$ manipulators, in which $m'_{nd} \in o(\hat{w},\calL_T')$ and it is the second most preferred man according to $\calF(\calL \cup \calL_T')$.
    %
    %
    According to Corollary~\ref{corrolary:successful_phase_1}, in the iteration in which $m_{nd} = m'_{nd}$, the algorithm successfully finishes phase $1$.
    Clearly, when there are $|R|$ manipulators, in the iteration in which $m_{nd} = m'_{nd}$, the algorithm will still successfully finish phase $1$.
    According to Lemma~\ref{lemma:woman_average_fromMax},
    $\max_{d \in D^{nd}}\set{s(d, \calL, \calL_R)} \leq q(D^{nd}) + \frac{3|D^{nd}|}{2} + 1$.
    %
    According to Lemma~\ref{lemma:woman_times_algorithm_returns_false},
    $
    q(D^{nd}) + \frac{3|D^{nd}|}{2} + 1 \leq s(m_{nd}, \calL, \calL_{R-2}) + \frac{3|D^{nd}|}{2} + 1.
    $
    Since $|D^{nd}| \leq k - 2$, then
    $
    s(m_{nd}, \calL, \calL_{R-2}) + \frac{3|D^{nd}|}{2} + 1 \leq s(m_{nd}, \calL, \calL_{R-2}) + \frac{3k}{2} - 2$.
    According to Lemma~\ref{lemma:additional_manipulator}, $s(m_{nd}, \calL, \calL_{R-2}) + \frac{3k}{2} - 2 \leq s(m_{nd}, \calL, \calL_{R}) + 3 - 2k + \frac{3k}{2} - 2$. Finally, since $B^{nd} \neq \emptyset$, then $D^{nd} \neq \emptyset$, therefore $k \geq 3$.
    Thus,
    $s(m_{nd}, \calL, \calL_{R}) + 3 - 2k + \frac{3k}{2} - 2 < s(m_{nd}, \calL, \calL_{R})$.
    Overall, 
    $\max_{d \in D^{nd}}\set{s(d, \calL, \calL_R)} < s(m_{nd}, \calL, \calL_R)$, and according to  Lemma~\ref{lemma:woman_times_algorithm_succeeds} the algorithm will find a preference profile that will make $w^*$ the match of $\hat{m}$ with $|R|$ manipulators.
    \end{enumerate}
\end{proof}
\end{document}